\documentclass[10pt,journal,twoside,compsoc]{IEEEtran}

\usepackage{amsmath}
\usepackage{graphicx}
\usepackage{txfonts}
\usepackage{multirow}
\usepackage{verbatim}
\usepackage{array}
\usepackage{subfig}
\usepackage[usenames, dvipsnames]{xcolor}
\usepackage[nocompress]{cite}
\usepackage[colorlinks, linkcolor= blue, urlcolor = blue, citecolor= blue]{hyperref} 
\usepackage{pifont}

 \usepackage[switch]{lineno}   
\usepackage{multirow}  
\newtheorem{definition}{Definition}   
\newtheorem{theorem}{Theorem}  
\newtheorem{corollary}{Corollary}   
\newenvironment{proof}{\begin{IEEEproof}}{\end{IEEEproof}}
\usepackage{amsmath}
\usepackage{arydshln} 
\usepackage{color} 
\usepackage{algorithm}
\usepackage{algorithmic}

\begin{document}

\title{Utility Mining Across Multi-Dimensional Sequences}

\author{Wensheng Gan,
	Jerry Chun-Wei Lin,
	Jiexiong Zhang,
	Hongzhi Yin,\\
	Philippe Fournier-Viger,
	Han-Chieh Chao,~\IEEEmembership{Senior,~Member}
	and Philip S. Yu,~\IEEEmembership{Fellow,~IEEE}
	\IEEEcompsocitemizethanks{
		\IEEEcompsocthanksitem Wensheng Gan is with Harbin Institute of Technology (Shenzhen), Shenzhen, China, and with University of Illinois at Chicago, IL, USA. Email: wsgan001@gmail.com
		
		\IEEEcompsocthanksitem Jerry Chun-Wei Lin is with the Western Norway University of Applied Sciences, Bergen, Norway. Email: jerrylin@ieee.org
		
		\IEEEcompsocthanksitem Jiexiong Zhang is with Harbin Institute of Technology (Shenzhen), Shenzhen, China. Email: jiexiong.zhang@foxmail.com
				
		\IEEEcompsocthanksitem Hongzhi Yin is with the University of Queensland, Brisbane, Australia. Email: h.yin1@uq.edu.au

		\IEEEcompsocthanksitem Philippe Fournier-Viger is with Harbin Institute of Technology (Shenzhen), Shenzhen, China. Email: philfv8@yahoo.com
				
		\IEEEcompsocthanksitem Han-Chieh Chao is with the National Dong Hwa University, Hualien, Taiwan. Email: hcc@ndhu.edu.tw
		
		\IEEEcompsocthanksitem Philip S. Yu is with University of Illinois at Chicago, IL, USA. Email: psyu@uic.edu}
	
}

\IEEEtitleabstractindextext{%

\begin{abstract}
 
 Knowledge extraction from database is the fundamental task in database and data mining community, which has been applied to a wide range of real-world applications and situations. Different from the support-based mining models, the utility-oriented mining framework integrates the utility theory to provide more informative and useful patterns. Time-dependent sequence data is commonly seen in real life. Sequence data has been widely utilized in many applications, such as analyzing sequential user behavior on the Web, influence maximization, route planning, and targeted marketing. Unfortunately, all the existing algorithms lose sight of the fact that the processed data not only contain rich features (e.g., occur quantity, risk, profit, etc.), but also may be associated with multi-dimensional auxiliary information, e.g., transaction sequence can be associated with purchaser profile information. In this paper, we first formulate the problem of utility mining across multi-dimensional sequences, and propose a novel framework named MDUS to extract \textit{\underline{M}}ulti-\textit{\underline{D}}imensional \textit{\underline{U}}tility-oriented \textit{\underline{S}}equential useful patterns. To the best of our knowledge, this is the first study that incorporates the time-dependent sequence-order, quantitative information, utility factor and auxiliary dimension. Two algorithms respectively named MDUS$ _\text{EM} $ and MDUS$ _\text{SD} $ are presented to address the formulated problem. The former algorithm is based on database transformation, and the later one performs pattern joins and a searching method to identify desired patterns across multi-dimensional sequences. Extensive experiments are carried on five real-life datasets and one synthetic dataset to show that the proposed algorithms can effectively and efficiently discover the useful knowledge from multi-dimensional sequential databases. Moreover, the MDUS framework can provide better insight, and it is more adaptable to real-life situations than the current existing models.
\end{abstract}

}
\maketitle

\section{Introduction}

\IEEEPARstart{T}{he} rise of ubiquitous data from e-commerce sites, social network, the Internet of Things (IoT) devices, and other services and devices, has generated new opportunities for data management and data analytic \cite{chen1996data}. Increasingly, the time-dependent sequence is one of the common types of these ubiquitous data.  Up to now, data mining \cite{chen1996data,agrawal1994fast}, especially sequence mining \cite{srikant1996mining,pei2004mining,fournier2017survey} which are  the fundamental technique for discovering useful knowledge between variables in large databases, have gained a lot of attention in both the public and the research communities. In general, some implicit factors such as the utility, interestingness, and risk of objects/patterns are commonly seen in many real-world scenarios \cite{ahmed2009efficient,tseng2013efficient}. Those information is actually important and helpful for data analytic. In some real-life applications, for example, consider products sold by an e-commerce website. The purchase quantities and unit profits of products are provided in transactions, as well as other information. In this situation, the frequency/co-occurrence is inappropriate to measure the importance of the product/object since much information may be missed and discarded. However, most of the pattern mining algorithms aim at extracting correlations among patterns mainly using the support measure or others \cite{geng2006interestingness,yin2012uspan}, which does not reflect business interests.

The above challenge motivates a field called utility-oriented mining \cite{gan2018survey}, which integrates the \textit{utility} theory \cite{marshall2005principles,mcconnell2009economics} and data mining technologies. Recently, the problem of data mining with Economic theory has been brought to attention \cite{yin2012uspan,zhao2017multi,liang2017top}.  Up to now, most studies of utility mining focus on extracting profitable patterns and useful knowledge from item-based data \cite{ahmed2009efficient,tseng2013efficient,liu2012mining}. While fewer studies \cite{yin2012uspan,liang2017top,alkan2015crom,wang2016efficiently,zihayat2017mining} address the utility-driven mining and analytic problem on sequence-based data which is more commonly seen in real-life scenarios. Sequence data has more rich information (e.g., time-dependent sequence-order) and more complex than item-based data. Data analytic with sequences has been widely utilized in many applications, such as analyzing sequential user behavior on the Web \cite{pei2004mining}, route planning \cite{liang2017top}, targeted marketing \cite{yin2012uspan}, and gene regulation from microarray \cite{zihayat2017mining}.

\begin{table*}[!htbp]
	\setlength{\abovecaptionskip}{5pt}
	\setlength{\belowcaptionskip}{5pt} 
	\caption{Purchase log with multi-dimensional features.}
	\centering
	
	\begin{tabular}{p{5pt}|ccccccc}
		\hline
		\textbf{SID} & \textbf{Time} & \textbf{Place}  & \textbf{Customer} & \textbf{Sex} & \textbf{Age} 	 & \textbf{Q-sequence} \\ \hline
		$ S_{1} $  & 5/2/2017 09:31 & Store &*** & Male    & Young    & $ < $[(\textit{a}:1) (\textit{c}:3)], [(\textit{a}:5) (\textit{c}:1) (\textit{e}:4)], [(\textit{c}:2)], [(\textit{b}:1)]$ > $  \\  \hdashline
		$ S_{2} $  & 5/2/2017 10:02  & Supermarket  &***  & Female  & Middle   & $ < $[(\textit{c}:1)], [(\textit{b}:4)], [(\textit{b}:9) (\textit{d}:8)], [(\textit{b}:9) (\textit{e}:6)]$ > $  \\  \hdashline
        $ S_{3} $  & 5/5/2017 12:25  & Supermarket  &***  & Male    & Young    & $ < $[(\textit{a}:10) (\textit{d}:5)]$ > $  \\ \hdashline
		$ S_{4} $   & 5/7/2017 10:30  &  Drugstore &***  & Male    & Young   & $ < $[(\textit{a}:3) (\textit{b}:4) (\textit{d}:2) (\textit{e}:6)], [(\textit{b}:3) (\textit{c}:2)]$ > $  \\ \hdashline
		$ S_{5} $   & 5/7/2017 16:58 &  Supermarket  &***  & Female  & Old      & $ < $[(\textit{e}:4)], [(\textit{d}:7)], [(\textit{c}:5)], [(\textit{a}:9) (\textit{b}:3) (\textit{c}:7) (\textit{d}:7)]$ > $  \\ \hline
	\end{tabular}
	\label{table:db}
\end{table*}

Up to now, most of the mining methods for sequence data aim at extracting correlations within a single attribute in a relation of order \cite{srikant1996mining,pei2004mining,zaki2001spade}. In real-life situations, sequence data may be associated with auxiliary information. For instance, a multi-dimensional purchase log is shown in Table \ref{table:db}. It contains 5 sequences/records, 5 items and several dimensions. Each sequence is associated with the purchaser profile information, such as \textit{Time}, \textit{Place}, \textit{Customer}, \textit{Set}, \textit{Age}, etc. This additional customer profile information provides rich multi-dimensional features to the sequence. Multi-dimensional sequence data is commonly seen in many real-life scenarios, such as market basket analysis, medical data for healthcare,  user behavior analysis, and event/route planning. For example, the medical data has many multi-dimensionality with rich information, and different disease (e.g., cough, asthma, pneumonia, dermatitis) should not be treated  equally since they have different risk/weight.  Traditional pattern mining algorithms cannot deal with those multi-dimensional features of the sequence. Pinto et al. \cite{pinto2001multi} first introduced the notion of multi-dimensionality in a sequence, and several algorithms were proposed to mine this type of knowledge \cite{yu2005mining,raissi2008mining,plantevit2010mining}. Unfortunately, all of the existing algorithms lose sight of the fact that the utility factor is quite important in market economics and data mining.

In the last two decades, database and data mining had been widely studied and applied to different domains, and some technologies were developed. While there are many advances in utility mining \cite{tseng2013efficient,yin2012uspan,liang2017top,zhao2017multi}, the characteristics of these data \cite{pinto2001multi} - complexity, volume, multi-dimensionality, etc. - still demand more efficient and effective techniques. Despite the encouraging applications of multi-dimensional data \cite{yu2005mining,raissi2008mining,plantevit2010mining}, none previous approach was dedicated to explore the multi-dimensionality of data by considering the time-dependent sequence-order, quantitative information, utility factor and auxiliary dimension.  We face several important questions: (1) What should be a sound definition of multi-dimensional utility-oriented sequential useful pattern? (2) How to formulate the problem of utility mining across multi-dimensional sequences? (3) How to evaluate these patterns efficiently? More specially, addressing this topic is not an easy task due to the following technical challenges.

First, it needs to overcome the large search space problem due to combinatorial explosion of sequences. However, the downward closure property does not hold for the utility of sequences. That is, the utility of a sequence may be higher than, equal to, or lower than that of its super/sub-sequences \cite{yin2012uspan,alkan2015crom,wang2016efficiently,zihayat2017mining,ahmed2010novel}. Thus, search space pruning techniques that rely on the downward closure property cannot be used to prune the search space of the addressed problem.

Second, in the multi-dimensional sequences, different items can occur simultaneously. This is substantially different and much more challenging than mining the patterns from a single sequence or multi-sequences without attribution dimensionality. Since items with different quantities and unit profits can occur simultaneously in any sequence record, the search space is much larger and the problem is much more challenging than previous studies.

Third, comparing to mining HUSPs from a normal dataset, mining HUSPs across multi-dimensional sequences has far more information to track and far greater complexity to manage. How to efficiently discover correct HUSPs to achieve utility maximization across multi-dimensional sequences is a challenging problem.

Therefore, efficient methods for utility-oriented sequence mining containing two or more quantitative dimensions is a non-trivial task. In this paper, we propose a new utility-oriented mining framework named \textbf{\underline{M}}ulti-\textbf{\underline{D}}imensional \textbf{\underline{U}}tility-oriented \textbf{\underline{S}}equential pattern mining (abbreviated as MDUS). In the designed framework, attributes are attached to transactions as dimensions. It incorporates the sequence-order information, utility factor and dimension information. In summary, this paper makes the following contributions:

\begin{enumerate}
	\item  \textbf{Dimension-based formulation.} We first formulate a novel MDUS framework to discover multi-dimensional sequential patterns, which maximize the total utility. As a sequence-specific mining framework, MDUS integrates the Economic utility theory, data mining technology and dimension information of data.  It provides better insight and is more adapted to real-life situations than existing studies.  

	\item  \textbf{More flexible and adaptable with tailored objective.} The newly formulated MDUS model can capture a partial or biased view of an individual domain analysis or application. When one or more dimensions of information is mined, it is specific to a particular domain or application type, and can give us more informative and useful patterns. 

	\item  \textbf{Two algorithms named MDUS$ _\text{EM} $ and MDUS$ _\text{SD} $ are presented.} The former algorithm is based on database transformation and the indexing utility-linked list. And the later one performs pattern joins and a new search approach to effectively identify high utility sequential patterns.

	\item  \textbf{Realistic and synthetic datasets were used in experiments.} Experiments conducted on six real-life datasets and one synthetic dataset show that the proposed algorithms have good performance  in terms of effectiveness and efficiency by processing multi-dimensional sequences. 
\end{enumerate}

The rest of this paper is organized as follows. In Section \ref{sec:preliminary}, basic preliminaries are introduced, and the problem of MDUS is formulated. The proposed two algorithms are described in Section \ref{sec:algorithm1} and Section \ref{sec:algorithm2}, respectively. Experimental results and analysis are presented in Section \ref{sec:experiments}. Related work is reviewed in Section \ref{sec:relatedwork}. At last, conclusions and future work are drawn in Section \ref{sec:conclusion}.

\section{Related Work}  
\label{sec:relatedwork}

This work is related to three areas: pattern mining from multi-dimensional sequence data, utility mining on itemset-based data, and utility mining on sequence data.

\subsection{Data Mining from Multi-dimensional Sequences} 
There are a lot of existing studies in the literature about finding the ``interesting'' sequential patterns from sequences when the minimum support threshold is given. Recently, Fournier-Viger et al. presented a comprehensive review of sequential pattern mining \cite{fournier2017survey}. However, most of the existing studies can not deal with the more realistic problem that mining interesting patterns from multi-dimensional sequences.

Pinto et al. \cite{pinto2001multi}  first addressed the problem to discover sequential patterns in a multi-dimensional space, where frequent sequences are associated with transactions to indicate their context. Three algorithms, named UniSeq, Seq-Dim, and Dim-Seq, were designed. Yu and Chen then proposed the AprioriMD and PrefixMDSpan \cite{yu2005mining} algorithms. To discover multi-dimensional sequential patterns in data streams, Ra{\"{\i}}ssi et al. \cite{raissi2008mining} proposed the MDSDS algorithm. It extracts the most specific multi-dimensional items from the prefix-tree in each window, and then uses the extracted items to mine frequent sequences using the PrefixSpan algorithm \cite{pei2004mining}. Plantevit et al. \cite{plantevit2010mining} then proposed a model for mining multi-dimensional and multi-level sequential patterns with multiple levels of granularities. Other algorithms have been proposed to process multi-dimensional data \cite{plantevit2010mining}. Unfortunately, all of these algorithms lose sight of the fact that the utility factor is quite important in real-word applications, such as market economics and utility maximization. Besides, the downward closure property using in these support-based algorithms does not hold for the utility of sequences. In other words, their search space pruning techniques cannot be used to prune the search space of the addressed MDUS problem.

\subsection{Utility Mining on Itemset-based Data}
Chan et al. \cite{chan2003mining} first incorporated the utility theory in a pattern mining task, and introduced the main concept of high-utility itemset mining (HUIM). Based on previous studies, Yao et al. \cite{yao2004foundational} then introduced a more restrictive definition of HUIM. Liu et al. \cite{liu2005two} developed a Two-Phase algorithm using level-wise Apriori-like mechanism. The proposed  transaction-weighted downward closure (TWDC) property can be hold although utility mining does not hold the downward closure property of Apriori \cite{agrawal1994fast}. Although this Two-Phase algorithm can completely discover the required results, it is still a Apriori-like model. To speed up the mining performance, several tree-based algorithms have been designed, such as IHUP \cite{ahmed2009efficient}, HUP-tree \cite{lin2011effective}, UP-Growth \cite{tseng2010up}, and UP-Growth$ ^{+} $ \cite{tseng2013efficient}. However, the tree-based models always consume a huge amount of memory, and need to generate a series of conditional subtrees. Therefore, a novel vertical data structure named utility-list was developed to efficiently discover the set of high-utility itemsets (HUIs), without multiple database scans. Besides, other more efficient algorithms (e.g., d$ ^{2} $HUP \cite{liu2016mining}, FHM \cite{fournier2014fhm}, and EFIM \cite{zida2015efim}) have been extensively studied.

Different from the above algorithms which are focused on improving the mining efficiency, there are many models and algorithms of utility mining \cite{gan2018survey} have been extensively studied to extend the mining effectiveness. For example, Tseng et al. studied the concise representation issue \cite{tseng2015efficient} and top-\textit{k} issue \cite{tseng2016efficient} for mining HUIs.  Several models are introduced to deal with different types of quantitative data, such as mining HUIs from uncertain data \cite{lin2016efficient}, temporal data \cite{lin2015efficient}, and dynamic data \cite{lin2015fast}. Up to now, there are also many other interesting issues for utility mining, such as on-shelf HUI \cite{lan2011discovery}, utility mining with discount strategies \cite{lin2016fast} or negative values \cite{lin2016fhn}, high-utility association rule \cite{mai2017lattice}, high utility occupancy pattern \cite{gan2018huopm},  and HUIM with multiple minimum utility thresholds \cite{2lin2016efficient}. Recently, Ryang at al. addressed the issues that mining HUIs from data stream \cite{ryang2016high}, and Yun et al. extended the utility mining technique for establishing manufacturing plan \cite{yun2017efficient}. A comprehensive survey literature of utility-oriented pattern mining can be referred to \cite{gan2018survey}.

\subsection{Utility Mining on Sequence Data}

Data mining on sequence, specially sequential pattern mining (e.g., \cite{chen1996data,srikant1996mining,pei2004mining}) provides more information to handle the order-based applications, such as behavior analysis, DNA sequence analysis, and weblog mining \cite{fournier2017survey,zaki2001spade}. It have been extensively studied, however it uses the frequency/support measure to mine desired sequences, which does not reflect business interests. Thus, a novel utility-oriented mining framework named high-utility sequential pattern mining (HUSPM) \cite{yin2012uspan,alkan2015crom,wang2016efficiently} is developed.

HUSPM considers ordered sequences and reveals the utilities of sequences. It has been widely studied and successfully applied to many applications, such as find useful episode rule from event sequences \cite{lin2015discovering,ao2018mining}, extract mobile sequential patterns \cite{shie2011mining}, route search with POI features \cite{liang2017top}, and discover high-utility gene regulation from microarray datasets \cite{zihayat2017mining}. Ahmed et al. \cite{ahmed2010novel} first designed a level-wise UL approach  and a pattern-growth US approach  for HUSPM. However, both UL and US can only handle simple sequences. Yin et al.~\cite{yin2012uspan} then introduced the formal framework of HUSPM and proposed an efficient USpan algorithm. Lan et al.~\cite{lan2014applying} then proposed a projection-based approach with a sequence-utility upper-bound with maximum utility measure. To improve the mining performance, Alkan et al.~\cite{alkan2015crom} proposed the Cumulated Rest of Match (CRoM) as the upper-bound value to early prune unpromising candidates. Two tight upper-bounds on utility, named prefix extension utility (PEU) and reduced sequence utility (RSU), are developed to speed up mining process \cite{wang2016efficiently}. Recently, two novel dynamic models, named  IncUSP-Miner \cite{wang2016incremental} and IncUSP-Miner+ \cite{wang2018incremental}, were proposed to incrementally mining high-utility sequential patterns. A comprehensive review of utility-oriented mining can be referred to \cite{2gan2018survey,gan2018survey}. Unfortunately, they are all one-dimensional pattern mining methods. None of them can be applied to identify high-utility patterns from multi-dimensional sequence data.

\section{Problem Formulation}\label{sec:preliminary}

\subsection{Utility Theory}
In Economics, \textit{utility} is a measure of the satisfaction or pleasure that a person gets from consuming a good or service \cite{marshall2005principles}. It is a basic building block of rational choice theory \cite{coleman1992rational}. A consumer's total utility for a given set of goods is the consumer's total amount of satisfaction experienced from consuming these goods as a whole. In the context of purchasing decisions, we assume that the consumer has accessed to a set of products, and each product has its own price value. In general, utility for a single product is a function of the consumed quantity. \textit{Utility} is subjective and difficult to quantify, and it typically obeys the \textit{Law Of Diminishing Marginal Utility} \cite{mcconnell2009economics}. In Economics, the utility that a consumer/seller has for a product can be decomposed into a set of utilities for each product characteristic. The utility maximization model provides insights into the effect of data mining, and it is useful for analyzing user behavior, influence maximization and targeted marketing \cite{zhao2017multi}.

\subsection{Utility Mining Across Multi-Dimensional Sequences}

In this section, we will first present the notations, as summarized in Table I. Then we will briefly introduce the conceptual paradigm of utility-oriented mining across multi-dimensional sequences. 

\begin{table}[!htbp]
	\centering
	\caption{Summary of symbols and their meanings}
	\label{table_Notation}
	\begin{tabular}{|c|l|}
		\hline
		\textbf{Symbol} & \textbf{Definition}  \\ \hline
		$I$ &  A finite set of distinct items in a database. \\ \hline
		
		$  q_{m} $ & An internal utility $  q_{m} $ (e.g., \textit{purchase quantity}). \\ \hline

		$pr(i_j)$ &   The predefined unit profit of an item $i_j$. \\ \hline
			
		\textit{v} & A \textit{quantitative itemset}  \textit{v} = [($ i_{1} $: $ q_{1} $) ($ i_{2} $: $ q_{2} $) $ \dots $ ($ i_{m} $: $  q_{m} $)]. \\ \hline
				
		$ s^{q} $	&  A \textit{quantitative sequence} is an ordered list. \\ \hline		
		
		$QSD$ &  A quantitative sequential database = \{\textit{S}$_{1}$, \textit{S}$_{2}$, $\ldots$, \textit{S$_{m}$}\}.  \\ \hline

     	$\delta$ &   A predefined minimum high-utility threshold. \\ \hline
		
		$u(QSD)$ &   Overall utility of a multi-dimensional \textit{QSD}. \\ \hline
				
		\textit{HUI} &   A high-utility itemset. \\ \hline	
		
		\textit{HUSP} &   A high-utility sequential pattern. \\ \hline
		
	    \textit{mdHUSP} &   	A multi-dimensional HUSP. \\ \hline
	 
		\textit{HUSPM} &   High-utility sequential pattern mining. \\ \hline
				
		length-$k$ pattern &  A pattern with the number of items/sequences (w.r.t. $k$). \\ \hline

        LS-tree   &  A lexicographic sequence. \\ \hline 
        
        UL-list    &  The utility-linked list structure. \\ \hline 
		
	\end{tabular}
\end{table}

Assume that items in an itemset \textit{v} (quantitative itemset) are unordered. For example, [\textit{a} \textit{b}] is the same as [\textit{b} \textit{a}]. Without loss of generality, we assume that items in an itemset (quantitative itemset) are listed in \textit{alphabetical} order. A \textit{sequence} is an ordered list of one or more itemsets without quantities, which is denoted as $ s $ = $<$$w_{1} $, $ w_{2} $, $\dots$, $ w_{m}$$>$. A \textit{quantitative sequence} is an ordered list of one or more quantitative itemsets, which is denoted as $ s^{q} $ = $<$$ v_{1} $, $ v_{2} $, $\dots$, $ v_{m}$$>$. For convenience, we use ``\textit{q}-" as an abbreviation of ``quantitative". Consequently, in the following, ``\textit{q}-itemset" denotes  an itemset with quantities, and ``itemset" denotes an itemset without quantities. Similarly, a ``\textit{q}-sequence" is a sequence with quantities, and ``sequence" denotes a sequence without quantities. For example, $ < $[(\textit{a}: 1) (\textit{b}: 2)], [(\textit{c}: 3)]$ > $ is a \textit{q}-sequence and $ < $[\textit{a b}], [\textit{c}]$ > $ is a sequence. [(\textit{a}: 1) (\textit{b}: 2)] is a \textit{q}-itemset and [\textit{a b}] is an itemset.

As a \textit{multi-dimensional sequence} \cite{pinto2001multi}, $ t $ has the form  $ (d_{1}, d_{2}, \dots, d_{m}, s) $, where $ d_{i} $ is the i-\textit{th} dimension value and $ s $ is a sequence. A multi-dimensional \textit{q}-sequence $ t^{q} $ has the form $ (d_{1}, d_{2}, \dots, d_{m}, s^{q}) $, where $ d_{i} $ is the i-\textit{th} dimension value and $ s^{q} $ is a \textit{q}-sequence. A multi-dimensional quantitative sequential database, defined as \textit{QSD} = \{$ S_{1} $, $ S_{2} $, $\dots$, $ S_{n} $\},  is a set of multi-dimensional \textit{q}-sequences. A multi-dimensional \textit{q}-sequence $ S_{i}\in QSD $ is also called a \textit{transaction}. For each transaction, a unique sequence identifier named \textit{SID} is attached to it. Besides, each item in \textit{QSD} is associated with an \textit{external utility} (e.g., representing its unit profit), which is denoted as $ pr(i_{j}) $. By taking into account the multi-dimensional sequences, we define the related concepts as follows.

As a running example, a simple multi-dimensional \textit{QSD} is shown in Table \ref{table:db}. It contains 5 sequences/records, 5 items and several dimensions. Assume that the user-specified \textit{utility-table} \{\textit{a}: \$4; \textit{b}: \$3; \textit{c}: \$5; \textit{d}: \$2; \textit{e}: \$1\} indicates the unit profits of all items in \textit{QSD}. It can be seen that \textit{S}$ _{1} $ is the first transaction, which can be viewed as a \textit{q}-sequence having three dimensions. The interpretation of $ S_{1}$ is that a young male doctor bought products \textit{a} and \textit{c}, followed by buying \textit{a}, \textit{c} and \textit{e}, then \textit{c}, and finally \textit{b}. Besides, [(\textit{a}:1) (\textit{c}:3)] is the first \textit{q}-itemset in $ S_{1}$. The quantity of an item (\textit{a}) in this \textit{q}-itemset is 1, and its utility is calculated as 1 $\times $ \$4 = \$4.

\begin{definition}
	Let $ u(i_{j}, v) $ denote the utility of an item ($ i_{j} $) in a \textit{q}-itemset \textit{v}, and it is defined as $u(i_{j}, v) = q(i_{j}, v)\times pr(i_{j})$, where $ q(i_{j}, v) $ is the quantity of ($ i_{j} $) in $ v $, and $ pr(i_{j}) $ is the profit of ($ i_{j} $). Let $ u(v) $ denote the utility of a \textit{q}-itemset $ v $, and it can be defined as $u(v) = \sum_{i_{j}\in v}u(i_{j}, v)$.
\end{definition} 

For example in Table \ref{table:db}, the utility of an item (\textit{c}) in the first \textit{q}-itemset of $ S_{1} $ is calculated as: \textit{u}(\textit{c}, [(\textit{a}:1) (\textit{c}:3)]) = \textit{q}(\textit{c}, [(\textit{a}:1) (\textit{c}:3)]) $\times pr(c)$ = 3 $\times$ \$5 = \$15. Thus, \textit{u}([(\textit{a}:1) (\textit{c}:3)]) = \textit{u}(\textit{a}, [(\textit{a}:1) (\textit{c}:3)]) + \textit{u}(\textit{c}, [(\textit{a}:1) (\textit{c}:3)]) = 1$\times$ \$4 + 3$\times$ \$5 = \$19.

\begin{definition}	
	Let  $ u(t^{q}) $ denote the utility of a multi-dimensional \textit{q}-sequence $ t^{q} $ = $ (d_{1}, d_{2}, \dots, d_{m}, s^{q}) $, and it is defined as the utility of the \textit{q}-sequence $ s^{q} $, such as $u(t^{q}) = u(s^{q}) = \sum_{v\in s^{q}}u(v)$.  Then the overall utility of a multi-dimensional \textit{QSD} can be calculated as $u(QSD) = \sum_{t^{q}\in QSD}u(t^{q})$. 
	
\end{definition}

For example in Table \ref{table:db}, $ u(S_{1}) $ = \textit{u}([(\textit{a}:1) (\textit{c}:3)]) + \textit{u}([(\textit{a}:5) (\textit{c}:1) (\textit{e}:4)]) + \textit{u}([(\textit{c}:2)]) + \textit{u}([(\textit{b}:1)]) = \$19 + \$29 + \$10 + \$3 = \$61, and the utility of \textit{QSD} is $ u$(\textit{QSD}) = $ u(S_{1}) $ + $ u(S_{2}) $ + $ u(S_{3}) $ + $ u(S_{4}) $ + $ u(S_{5}) $ = \$61 + \$93 + \$50 + \$53 + \$137 = \$394.

\begin{definition}
	An itemset (\textit{q}-itemset) containing \textit{k} items is called a \textit{k}-itemset (\textit{k}-\textit{q}-itemset). A sequence (\textit{q}-sequence) containing \textit{k} items is said to be a \textit{k}-sequence (\textit{k}-\textit{q}-sequence). 
\end{definition}

\begin{definition}
	Given two itemsets $ w $ and $ w' $, $ w $ is said to be contained in $ w' $, denoted as $ w\subseteq w' $, if $ w $ is a subset of $ w' $ or $ w = w' $. Given two \textit{q}-itemsets $ v $ and $ v' $, $ v $ is said to be contained in $ v' $ as $ v\subseteq v' $ if for any item in $ v $, there exists the same item having the same quantity in $ v' $.
\end{definition}

For example, the \textit{q}-sequence of $ S_{1} $ is a 7-\textit{q}-sequence and its first \textit{q}-itemset is a 2-\textit{q}-itemset. An itemset [\textit{a b}] is contained in the itemset [\textit{a b c}]. The \textit{q}-itemset [(\textit{a}:1) (\textit{c}:3)] is contained in [\textbf{(\textit{a}:1)} (\textit{b}:1) \textbf{(\textit{c}:3)}] and [\textbf{(\textit{a}:1)} \textbf{(\textit{c}:3)} (\textit{e}:2)], but is not contained in [(\textit{a}:1) (\textit{b}:3) (\textit{c}:1)] and [(\textit{a}:4) (\textit{c}:3) (\textit{d}:4)].

\begin{definition}
	Given two sequences $ s $ = $<$$w_{1}, w_{2}, \dots, w_{m}$$>$ and $ s' $ = $<$$w'_{1}, w'_{2}, \dots, w'_{m'}$$>$, $ s $ is said to be contained in $ s' $, denoted as $ s\subseteq s'$, if there exists an integer sequence $ 1\leq k_{1}\leq k_{2}\leq\dots\leq m' $ such that $ w_{j}\subseteq w'_{k_{j}} $ for $ 1\leq j\leq m $. 
	Given two \textit{q}-sequences $ s^{q} $ = $<$$v_{1}, v_{2}, \dots, v_{m}$$>$ and $ s^{q'} $ = $<$$v'_{1}, v'_{2}, \dots, v'_{m'}$$>$, $ s^{q} $ is said to be contained in $s^{q'}$, denoted  as $ s^{q}\subseteq s^{q'} $, if there exists an integer sequence $ 1\leq k_{1}\leq k_{2} \leq\dots\leq m' $ such that $ v_{j}\subseteq v'_{k_{j}} $ for $ 1\leq j \leq m $.
\end{definition}	

For example, the sequence $<$[\textit{a}], [\textit{c}]$>$ is contained in $<$[\textit{a}], [\textit{c}], [\textit{d}]$>$. And $<$[(\textit{a}:1)], [(\textit{c}:1)]$>$ is contained in the \textit{q}-sequence of $ S_{1} $, but $<$[(\textit{a}:1)], [\textit{c}:3]$>$ is not contained in the \textit{q}-sequence of $ S_{1} $.

\subsection{Modeling Sequence Dimensionality Requirement}

To address the personalized sequence dimensionality requirement, we modeling a new framework by considering utility factor, time-dependent order, event/sequence and semantic dimensionality which embedding in multi-dimensional sequences. We aim to provide a general solution to the utility mining across multi-dimensional sequences, such as targeted marketing, influence maximization, Web data analytic, and the route search problem.

\begin{definition}
	Given two multi-dimensional sequences $ t $ = $ (d_{1}$, $d_{2}$, $\dots$, $d_{m}, s) $ and $ t' $ = $ (d'_{1}$, $d'_{2}$, $\dots$, $d'_{m}, s') $, $ t $ is said to be contained in $ t' $, and denoted as $ t\subseteq t $' if $ d_{i} $ = $ d'_{i} $ or $ d'_{i} $ = * for $ 1 \leq\ i \leq m $ and $ s \subseteq s' $.
	Given two multi-dimensional \textit{q}-sequences $ t^{q} $ = $ (d_{1}$, $d_{2}$, $\dots$, $d_{m}, s^{q}) $ and $ t^{q'} $ = $ (d'_{1}$, $d'_{2}$, $\dots$, $d'_{m}$, $s^{q'}) $, $ t^{q} $ is said to be contained in $ t^{q'} $ and denoted as $ t^{q}\subseteq t^{q'} $ if $ d_{i} $ = $ d'_{i} $ or $ d'_{i} $ = * for $ 1 \leq\ i \leq m $ and $ s^{q} \subseteq s^{q'} $. Note that * indicates that $ d'_{i} $  can take any value.
\end{definition}	

Consider the sequence (\textit{Male}, \textit{Young}, \textit{Doctor}, $ < $\textit{a}, \textit{b}$ > $), it is contained in (\textit{Male}, \textit{Young}, *, $ < $\textit{a}, \textit{b}$ > $) and (\textit{Male}, \textit{Young}, \textit{Doctor}, $ < $\textit{a}, \textit{b}, \textit{c}$ > $). However, it is not contained in (*, \textit{Middle}, \textit{Doctor}, $ < $\textit{a}, \textit{b}$ > $) and (\textit{Male}, *, \textit{Doctor}, $ < $\textit{a}, \textit{c}, \textit{d}$ > $).

\begin{definition}
	Given a \textit{q}-sequence $ s^{q} $ = $<$$v_{1}, v_{2}, \dots, v_{m}$$>$ and a sequence $ s $ = $<$$w_{1}, w_{2}, \dots, w_{m'}$$>$ if $ m $ = $ m' $ and the items in $ v_{k} $ are the same as the items in $ w_{k} $ for $ 1\leq k\leq m $. Then, $ s $ matches $ s^{q} $, which is denoted as $ s\sim s^{q} $.	
	Given a multi-dimensional \textit{q}-sequence $ t^{q} $ = $ (d_{1}, d_{2}, \dots, d_{m}, s^{q}) $ and a multi-dimensional sequence $ t $ = $ (d'_{1}, d'_{2}, \dots, d'_{m}, s) $, if $ d_{i}=d'_{i} $ for $ 1\leq i\leq m $ and $ s $ matches $ s^{q} $, then $ t $ matches $ t^{q} $ which is denoted as $ t\sim t^{q} $.	
\end{definition}

For example, (\textit{Male}, \textit{Young}, \textit{Doctor}, $<$[\textit{a c}], [\textit{a c e}], [\textit{c}], [\textit{b}]$>$) matches $ S_{1} $. Moreover, $<$[\textit{a c}], [\textit{a c e}], [\textit{c}], [\textit{b}]$>$ matches the \textit{q}-sequence of $ S_{1} $. Note that a sequence may have multiple matches in a \textit{q}-sequence.  For example, $<$[\textit{a}], [\textit{c}]$>$ has three matches that are $<$[\textit{a}:1], [\textit{c}:1]$>$, $<$[\textit{a}:1], [\textit{c}:2]$>$ and $<$[\textit{a}:5], [\textit{c}:2]$>$ in the \textit{q}-sequence of $ S_{1} $. This is one of the reasons why high-utility sequential pattern mining is more challenging.

\begin{definition}
	Let $ u(t, t^{q}) $ and $ u(t) $ denote the overall utility of a multi-dimensional sequence $ t $ in a multi-dimensional \textit{q}-sequence $ t^{q} $ and a \textit{QSD}, respectively. They can be respectively calculated as $u(t, t^{q})$ = $max\{u(t^{q'})$$|t \sim t^{q'}$$ \wedge t^{q'} \subseteq t^{q}\}$ and  $ u(t)$ = $\sum_{t^{q}\in QSD} \{u(t,t^{q}) | t \subseteq t^{q}\}$. Note that for convenience, we use $ t \subseteq t^{q} $ to indicate that $ t \sim t^{q'} \wedge t^{q'} \subseteq t^{q} $. 
\end{definition}

For example in Table \ref{table:db}, \textit{u}((\textit{Male}, \textit{Young}, \textit{Doctor}, $ < $[\textit{a}], [\textit{c}]$ > $), \textit{S}$_{1} $)  = $ max$\{\textit{u}($<$[\textit{a}:1], [\textit{c}:1]$>$), \textit{u}($<$[\textit{a}:1], [\textit{c}:2]$>$), \textit{u}($<$[\textit{a}:5], [\textit{c}:2]$>$)\} = $ max $\{\$9, \$14, \$30\} = \$30. However, \textit{u}((\textit{Female}, \textit{Young}, \textit{Doctor}, $ < $[\textit{a}], [\textit{c}]$ > $), \textit{S}$_{1} $) = \$0 since there are no matches of this pattern in \textit{S}$ _{1} $. And \textit{u}((\textit{Male}, \textit{Young}, *, $<$[\textit{a}], [\textit{c}]$>$)) = \textit{u}((\textit{Male}, \textit{Young}, *, $<$[\textit{a}], [\textit{c}]$>$), $S_{1}$) + \textit{u}((\textit{Male}, \textit{Young}, *, $<$[\textit{a}], [\textit{c}]$>$), $S_{4}$) = \$30 + \$22 = \$52. From this example, we can see that a (multi-dimensional) sequence has multiple utility values in a (multi-dimensional) \textit{q}-sequence, which is noticeably different from frequent sequential pattern mining.

\begin{definition}
	Given a multi-dimensional \textit{QSD}, a multi-dimensional sequence $ t $ in \textit{QSD} is said to be a multi-dimensional high-utility sequential pattern (abbreviated as \textit{mdHUSP}) if its utility satisfies $u(t)\geq u(QSD) \times \delta$, where $\delta$ is a user-defined minimum utility threshold based on the prior knowledge. 
\end{definition}

For example in Table \ref{table:db}, if $\delta$ is set to 0.1, (\textit{Male}, \textit{Young}, *, $<$[\textit{a}], [\textit{c}]$>$) is a mdHUSP since \textit{u}((\textit{Male}, \textit{Young}, *, $<$[\textit{a}], [\textit{c}]$>$)) (= \$52) $ > $ $ u(QSD) \times \theta $ (= \$394 $ \times $ 0.1). Based on the above definitions, the problem of utility mining across multi-dimensional sequences by discovering \textbf{\underline{M}}ulti-\textbf{\underline{D}}imensional \textbf{\underline{U}}tility-oriented \textbf{\underline{S}}equential patterns (abbreviated as MDUS) can be defined as follows.

\textbf{Problem Statement:} Given a multi-dimensional \textit{QSD} and a user-defined minimum utility threshold $\delta$, MDUS is the task for discovering all mdHUSPs whose overall utilities across multi-dimensional sequences in \textit{QSD} are no less than $u(QSD) \times \delta$. By utilizing the dimensions of MDUS, we can extract general or specific patterns to achieve utility maximization and to overcome the problems of excessive granularities and high utilities.

\section{Proposed MDUS$ _\text{EM} $ Algorithm} \label{sec:algorithm1}

Based on the above concepts, in this section, the first algorithm named MDUS$ _\text{EM} $ ($E$ refers to the \underline{E}quivalence transformation part, $M$ refers to the \underline{M}ining part) is proposed to discover multi-dimensional HUSPs that achieve utility mining across multi-dimensional sequences.

\subsection{Transformation of Database}

To transform database, we should first compare the two database formats. Based on the previous definitions, it can be seen that a multi-dimensional \textit{q}-sequence can be viewed as a \textit{q}-sequence with several dimensions. For example, the first transaction \textit{S}$ _{1} $ is (\textit{Male}, \textit{Young}, \textit{Doctor}, $ < $[(\textit{a}:1) (\textit{c}:3)], [(\textit{a}:5) (\textit{c}:1) (\textit{e}:4)], [(\textit{c}:2)], [(\textit{b}:1)]$ > $), which can be divided into two parts: a \textit{q}-sequence and three dimensions. Thus, for a multi-dimensional \textit{q}-sequence $ t^{q} $ = $ (d_{1}, d_{2}, \dots, d_{m}, s^{q}) $, the transformation considers the dimensions of $ t^{q} $ as an itemset and adds them to $ s^{q} $. Since the dimension information has no utility values, the utility of each dimension in an itemset is 0. For example, transforming the multi-dimensional \textit{q}-sequence  \textit{S}$ _{1} $ from Table \ref{table:db} results in this \textit{q}-sequence:  $<$[(\textit{a}:1) (\textit{c}:3)], [(\textit{a}:5) (\textit{c}:1) (\textit{e}:4)], [(\textit{c}:2)], [(\textit{b}:1)], [(\textit{Male}:0) (\textit{Young}:0) (\textit{Doctor}:0)]$>$. Table \ref{table:transformed} shows the database obtained by transforming the original database of Table \ref{table:db}.

\begin{table}[!htbp]
	\setlength{\abovecaptionskip}{0pt}
	\setlength{\belowcaptionskip}{0pt} 
	\caption{A transformed multi-dimensional \textit{QSD}.}
	\centering
	\begin{tabular}{c|c}
		\hline
		\textbf{SID} &  \textbf{Q-sequence} \\ \hline
		$ S_{1} $ & $<$ [(Male:0)(Young:0)(Doctor:0)], [(\textit{a}:1) (\textit{c}:3)], [(\textit{a}:5) (\textit{c}:1)\\  &   (\textit{e}:4)], [(\textit{c}:2)], [(\textit{b}:1)]$>$  \\ \hdashline
		$ S_{2} $ & $<$ [(Female:0)(Middle:0)(Lawyer:0)], [(\textit{c}:1)], [(\textit{b}:4)], [(\textit{b}:9) \\  &  (\textit{d}:8)], [(\textit{b}:9) (\textit{e}:6)]$>$  \\ \hdashline
		$ S_{3} $ & $<$ [(Male:0)(Child:0)(Driver:0)], [(\textit{a}:10) (\textit{d}:5)]$>$  \\ \hdashline
		$ S_{4} $ & $<$ [(Male:0)(Young:0)(Writer:0)], [(\textit{a}:3) (\textit{b}:4) (\textit{d}:2) (\textit{e}:6)] \\  &  [(\textit{b}:3) (\textit{c}:2)]$>$  \\ \hdashline
		$ S_{5} $ & $<$ [(Female:0)(Old:0)(Artist:0)], [(\textit{e}:4)], [(\textit{d}:7)], [(\textit{c}:5)], \\  &  [(\textit{a}:9) (\textit{b}:3) (\textit{c}:7) (\textit{d}:7)]$>$  \\ \hline
	\end{tabular}
	\label{table:transformed}
\end{table}

After all transactions in a multi-dimensional \textit{QSD} have been transformed into quantitative sequences, a search procedure can be applied on the transformed database. Note that the discovered results are the  HUSPs, and they should be converted back to multi-dimensional HUSPs. In this process, dimensions appearing in HUSPs are then considered as dimension values. Moreover, for dimensions not appearing in a pattern, the value * is assigned as dimension value. For example, if $<$[\textit{a} \textit{c}] [\textit{b}] [\textit{Young} \textit{Doctor}]$>$ is a HUSP found in the transformed database, and the corresponding multi-dimensional HUSP is (*, \textit{Young}, \textit{Doctor}, $<$[\textit{a} \textit{c}] [\textit{b}]$>$). This pattern can be interpreted as young doctors prefer to buy products \textit{a}, \textit{c} and \textit{b} in that order.

\subsection{Proposed MDUS$_\text{EM} $ Algorithm}

Based on the above operations, the MDUS$ _\text{EM} $ algorithm is presented below. It first transforms a multi-dimensional \textit{QSD} into a traditional \textit{QSD}. Then, a search procedure is applied to the transformed database. After that, the MDUS$ _\text{EM} $ approach obtains multi-dimensional HUSPs from the mined results. The process of MDUS$ _\text{EM} $ is a non-trivial task. To ensure the correctness and completeness of the mining results, we have to ensure the equivalence of the original problem and the transformed problem. We address this issue by proposing the following theorem.

\begin{theorem}[\textbf{Equivalence Property}]
	\label{theorem-em-equal}
	Given a multi-dimensional \textit{QSD} and a minimum utility threshold $ \delta $. Let mdHUSPs be the final set of multi-dimensional high-utility sequential patterns, and HUSPs be the set of high-utility sequential patterns in the transformed database. Then, mdHUSPs = HUSPs is obtained.
\end{theorem} 

\begin{proof}
	Let $ t $ be a mdHUSP, and $ t' $ be the corresponding HUSP $ t $. Based on the definition of the utility of a multi-dimensional sequential pattern, we can obtain that $ u(t)$ = $\sum_{t^{q}\in QSD} \{u(t,t^{q})$ $| t \subseteq t^{q}\}$ = $\sum_{t^{q}\in QSD} max\{u(t^{q'})|t \sim t^{q'} \wedge t^{q'}$ $\subseteq t^{q}\}$ = $\sum_{t^{q}\in QSD} max\{u(s^{q'})|t \sim t^{q'}$ $\wedge t^{q'} \subseteq t^{q} \wedge s^{q'} \in t^{q'} \} $, where $ s^{q'} $ is a \textit{q}-sequence in the multi-dimensional \textit{q}-sequence $ t^{q'} $. From these results, it can be seen that the utility of a multi-dimensional sequential pattern is the utility of $ t $ in the \textit{q}-sequences of transactions. In other words, calculating the utility is independent of the dimensions. 
	
	In addition, when transactions are transformed into $q$-sequences, only the dimensions in transactions are changed. This change does not modify the utility of $ t $ since the quantities of these dimensions are 0. Thus, the utility of $ t $ does not change after the database transformation, which means that $ u(t) = u(t') $. As a result, $ u(t') = u(t) \geq u(QSD) \times \delta $. Hence, we can have that $ t' \in $ HUSPs and mdHUSPs $ \subseteq $ HUSPs.
	Let $ t' $ be a high-utility sequence in HUSPs and $ t $ be the original sequence of $ t' $, we can get that $ u(t) = u(t') \geq u(QSD) \times \delta $. Thus, $ t $ $ \in $ mdHUSPs and HUSPs $ \subseteq $ mdHUSPs.
	In summary, we can obtain that mdHUSPs = HUSPs.

\end{proof}

\begin{definition}[\textbf{\textit{I}-\textit{Concatenation}} and \textbf{\textit{S}-\textit{Concatenation}}]
	Given a sequence $ t $ and an item $ i_{j} $, the \textit{I-Concatenation} of $ t $ with $ i_{j} $ consists of appending $ i_{j} $ to the last itemset of $ t $, denoted as $<$$t \oplus i_{j}$$>$$_{I-Concatenation}$. 
	The \textit{S-Concatenation} of $ t $ with an item $ i_{j} $ consists of adding $ i_{j} $ to a new itemset appended after the last itemset of $ t $, denoted as $<$$t \oplus i_{j}$$>$$_{S-Concatenation}$.
\end{definition}

For example, given a sequence $ t $ = $<$[$a$], [$b$]$>$ and a new item $(c)$, $<$$t \oplus c$$>$$_{I-Concatenation}$ = $<$[\textit{a}],[\textit{bc}]$>$ and $<$$t \oplus c$$>$$_{S-Concatenation}$ = $<$[\textit{a}],[\textit{b}],[\textit{c}]$>$. 
Based on the previous definitions, it follows that the number of itemsets in $ t $ does not change after performing an \textit{I}-\textit{Concatenation}, while performing an \textit{S-Concatenation} increases the number of itemsets in $ t $  by one. The search space of the addressed MDUS problem can be presented as an extended Set-enumeration tree, as shown in Fig. \ref{fig:tree}. It shows a built partial lexicographic sequence-tree (LS-tree) \cite{wang2016efficiently,yin2012uspan} based on Table \ref{table:db}. 

\begin{figure}[!htbp]
	\setlength{\abovecaptionskip}{0pt}
	\setlength{\belowcaptionskip}{0pt} 		
	\centering
	\includegraphics[width=3.6in]{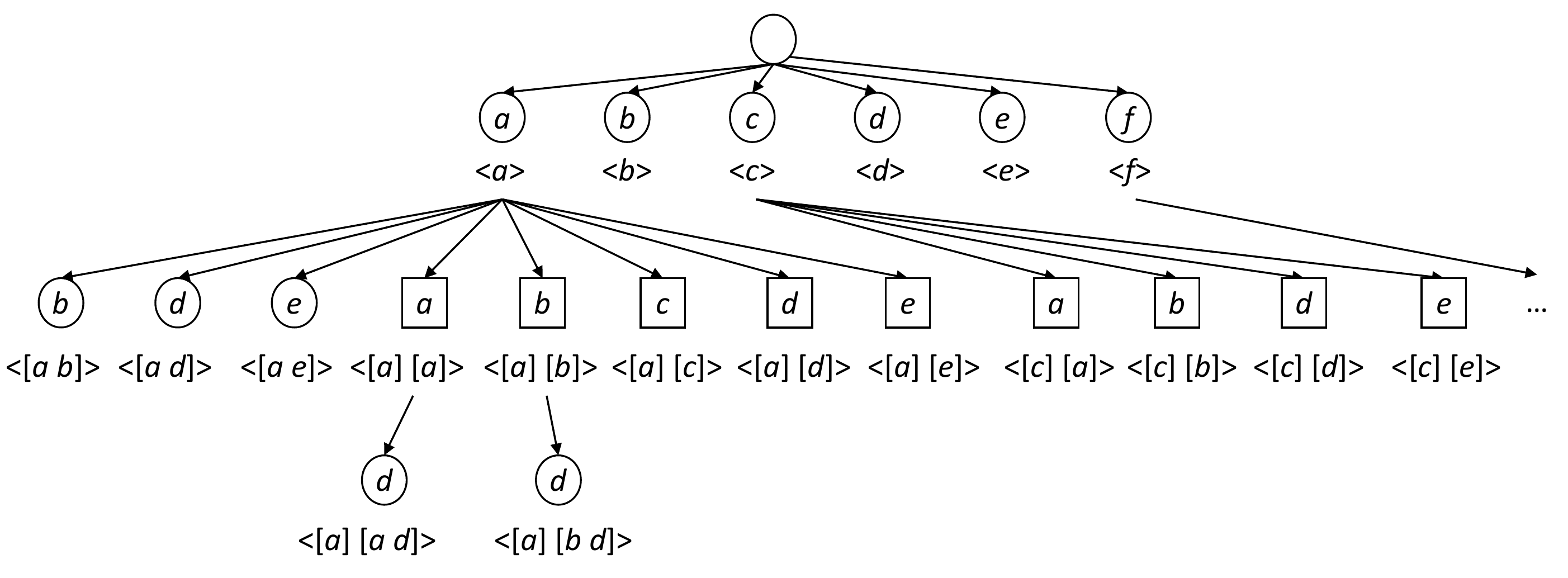}
	\caption{A lexicographic sequence (LS)-tree.}
	\label{fig:tree}
\end{figure}

Based on the two operations, all candidates of the search space can be generated and checked for mining mdHUSPs. To efficiently deal with the high computational complexity of this problem, some upper bounds on utility (e.g., SWU \cite{yin2012uspan}, prefix extension utility (PEU) \cite{wang2016efficiently}, and reduced sequence utility (RSU) \cite{wang2016efficiently}) are adopted in our model. However, it needs to scan the database many times, which would result in long execution times because there are often multiple matches in a sequence. To handle this situation, the compact structure named utility-linked (UL)-list \cite{lin2017high} is adopted here. The UL-list contains two arrays, \textit{Header Table} and \textit{UP (utility and position) Information}. Details are described below. \textbf{1) Header Table}. It stores  a set of distinct items with their first occurrence positions in the transformed transaction. \textbf{2) UP Information}. In terms of utility and position information of each sequence, each element respectively stores the \textit{\underline{item name}}, the \textit{\underline{utility of the item}}, the \textit{\underline{remaining utility of the item}}, and the \textit{\underline{next position of the item}}. For each node in the LS-tree, transactions containing this node (sequence) are transformed into a utility-linked (UL)-list and attached to the related database of this node. The utilities and upper-bounds of the candidates can be easily calculated from the  database using the UL-list structure.

\subsection{Framework Overview of MDUS$ _\text{EM} $}

Based on Theorem \ref{theorem-em-equal}, we can guarantee the  completeness and correctness of the MDUS$ _\text{EM} $ algorithm. The theorem ensures the equivalence between the original problem and the transformed problem. The framework overview of MDUS$ _\text{EM} $ is given in Algorithm \ref{alllgo}. In general, the MDUS$ _\text{EM} $ algorithm consists of three phases: 1) database transformation, 2) extract HUSPs by using the \textbf{search} procedure, and 3) transformation of the found HUSPs. It first scans and transforms each transaction from the original multi-dimensional database into a \textit{q}-sequence. The function \textbf{transform} is a transformation procedure, which identifies all dimensions in transactions and puts them into the \textit{q}-sequence as an itemset. The quantities of these dimensions are 0.  After that, MDUS$ _\text{EM} $ invokes a \textbf{search} procedure (c.f. Algorithm \ref{HUSPM-Miner}) to mine the processed database. Then, it transforms the results into multi-dimensional HUSPs. Here the mining procedure can adopt an existing HUSPM algorithm, but we improve the efficiency by using the UL-list \cite{lin2017high} structure and \textbf{PGrowth} searching procedure. Details of \textbf{transform} and \textbf{PGrowth} are skipped here due to space limitation. Note that the most difficult part is to ensure the equivalence between the original problem and the transformed problem, and the complexity of MDUS$ _\text{EM} $ depends on the employed \textbf{search} procedure or HUSPM algorithm.

\begin{algorithm}
	\caption{MDUS$ _\text{EM} $}
	\label{alllgo}
	\begin{algorithmic}[1]
		\REQUIRE {\textit{QSD}, a multi-dimensional quantitative sequential database; \textit{utable}, a utility table containing the unit profit of each item; $ \delta $, a minimum utility threshold.} 
		\ENSURE {The set of \textit{mdHUSPs}.}
		
		\STATE $ QSD' \gets \varnothing $;
		\FOR {each $ t^{q}\in QSD $}
			\STATE $ s \gets transform(t^q) $;
			\STATE $ QSD' \gets QSD' \cup s $;
		\ENDFOR
		\STATE \textit{HUSPs} $\gets$ HUSPM$(QSD', utable, \delta) $;
		\STATE \textit{mdHUSPs} $ \gets \varnothing $;
		\FOR {each $ s\in $ \textit{HUSPs}}
			\STATE $ t\gets transform'(s) $;
			\STATE \textit{mdHUSPs} $ \gets $ \textit{mdHUSPs} $\cup t $;
		\ENDFOR
		\STATE \textbf{return} \textit{mdHUSPs}
	\end{algorithmic}
\end{algorithm}

\begin{algorithm}
	\caption{search}
	\label{HUSPM-Miner}	
	\begin{algorithmic}[1]
		\REQUIRE {\textit{QSD'}, a transformed \textit{QSD}; \textit{utable}, a utility table; $\delta$, a minimum utility threshold.} 
		\ENSURE {The set of $ HUSPs $.}

		\STATE scan $ QSD' $ to: 1). calculate $u(s)$ for each $s\in QSD' $ and calculate $ u(QSD') $; 2). build the UL-list of each $s\in QSD'$;
		\STATE $HUSPs \gets \emptyset $;
		\FOR {each $ i_{j}\in QSD' $}
		
			\STATE \textit{PD}($ < $$ i_{j} $$ > $)$\gets$\{the UL-list of $s|$$<$$i_{j}$$>$$\subseteq s\wedge s\in QSD'\}$;
			\STATE	calculate $SWU$($<$$i_{j}$$>$) and $u$($<$$i_{j}$$>$);
			\IF {$SWU$($<$$i_{j}$$>$)$ \geq \delta $ $ \times u(QSD') $}
				\IF {$u$($<$$i_{j}$$>$)$\geq \delta \times u(QSD')$)}
				\STATE $HUSPs$$\gets$$HUSPs$$\cup$$<$$i_{j}$$>$;
			\ENDIF
			
			\ENDIF
		\ENDFOR
		
		\STATE $HUSPs$$\gets$$HUSPs$$\cup$$<$$i_{j}$$>$;
		\STATE \textbf{call PGrowth}($ < $$ i_{j} $$ > $, \textit{PD}($ < $$ i_{j} $$ > $), \textit{HUSPs});
		\STATE \textbf{return} \textit{HUSPs}
	\end{algorithmic}
\end{algorithm}


\textbf{\underline{Complexity analysis}}. Although the proposed MDUS$ _\text{EM} $ algorithm can effectively discover the complete set of multi-dimensional HUSPs, it is inefficient since all the existing HUSPM algorithms are not specifically designed for the addressed problem. A large number of dimensions can result in transformed transactions that are very long. Processing long transactions results in performing numerous iterations and generating many candidates. More especially, the complexity of the \textbf{search} procedure  for discovering HUSPs mainly lies in searching the LS-tree w.r.t. the maximal number of length-1 sequences.  As a result, the search space of the transformed problem can be very large, and the runtime and memory consumption may dramatically increase. To address this issue, the second algorithm is presented below to efficiently mine multi-dimensional HUSPs.

\section{Proposed MDUS$ _\text{SD} $ Algorithm} \label{sec:algorithm2}

In this section, we present a more efficiently algorithm named MDUS$ _\text{SD} $ ($S$ refers to the \underline{\textit{S}}equential part, $D$ refers to the \underline{\textit{D}}imensional part), and the details are described below.

\subsection{Pattern Join}

The main idea of MDUS$ _\text{SD} $ is to mine the sequential part of the database using HUSPM techniques, and mine the dimensional part of the database using a novel DHUI-Miner algorithm, respectively. Then, multi-dimensional HUSPs are obtained from these intermediate results by performing a \textit{pattern join} operation. 

The MDUS$ _\text{SD} $ algorithm first discovers the set of HUSPs without dimensional information. For each HUSP, it finds the dimensions information from the dimensional part of the database using the DHUI-Miner algorithm (it will be described in next subsection). Then MDUS$ _\text{SD} $ combines the sequential patterns and dimensions to return multi-dimensional HUSPs. Note that MDUS$ _\text{SD} $ handles the sequential part and dimensional part of the database separately, which eliminates the problem of processing very long transactions. 

\begin{theorem}[\textbf{Equivalence Property}]
	\label{theorem-sd-equal}
	Given a multi-dimensional \textit{QSD} and a minimum utility threshold $ \delta $, and let $ t = (d_{1}, d_{2}, \dots, d_{m}, s) $ be a multi-dimensional high-utility sequential pattern of \textit{QSD}, then $ s $ is a high-utility sequential pattern in the sequential part of \textit{QSD}.
\end{theorem} 

\begin{proof}
	From Theorem \ref{theorem-em-equal}, we can get that $ s $ is equal to $ ts = (*, *, \dots, s) $, and their utilities are the same. Based on the definition of multi-dimensional sequential pattern, $ ts $ contains $ t $. Thus, $ u(s) = u(ts) \geq u(t) \geq u(QSD) \times \delta $. Hence, if $ t $ is a multi-dimensional HUSP of \textit{QSD}, $ s $ is a HUSP in the sequential part of the database \textit{QSD}.
\end{proof}

\begin{corollary}		
 \label{corollary-sd-equal}
	\rm Given a multi-dimensional \textit{QSD} and the minimum utility threshold $ \delta $, if a sequence is not a HUSP in the sequential part of \textit{QSD}, then those multi-dimensional patterns containing this sequence are not the desired multi-dimensional HUSPs.
\end{corollary}

Theorem \ref{theorem-sd-equal} and Corollary \ref{corollary-sd-equal} ensure the completeness of the MDUS$ _\text{SD} $ algorithm. If a sequence $ s $ is not a HUSP in the sequential part of \textit{QSD}, the MDUS$ _\text{SD} $ algorithm will not search the dimensional part of the database of $ s $. Using Corollary \ref{corollary-sd-equal} can reduce the number of candidates and improve the efficiency of the algorithm. To reduce the complexity, we propose the \textit{utility-list} structure and downward closure property for the DHUI-Miner algorithm, which can mine the dimensional part of the database. Details are described as below.

\subsection{Proposed DHUI-Miner Algorithm}

To mine the dimensional part of the database, a novel DHUI-Miner algorithm is proposed. In the first phase, the MDUS$ _\text{SD} $ algorithm  discovers HUSPs in the sequential part of the database using the \textbf{search} procedure (c.f.  Algorithm	\ref{HUSPM-Miner}). Different from HUSPM, MDUS$ _\text{SD} $ also records the utilities of each high-utility sequential pattern (HUSP) in transactions. This utility information and the original database are given as the input to the DHUI-Miner algorithm.

Consider the running example, MDUS$ _\text{SD} $ first discovers high-utility sequential patterns. In the found results, $u(<$$a$$>)$ = 108. If $<$$a$$>$ is a HUSP, MDUS$ _\text{SD} $ will call the DHUI-Miner approach to find dimensions for $<$$a$$>$. Table \ref{table:dimension} is the input database of DHUI-Miner. This database is obtained by transforming the original database, and it consists of the identifiers of transactions (\textit{SID}), dimensions and the utilities of patterns of each transaction (\textit{TU}).

\begin{table}[!htbp]
	\setlength{\abovecaptionskip}{0pt}
	\setlength{\belowcaptionskip}{0pt} 
	\caption{The dimensional database of $<$$a$$>$.}
	\centering
	\begin{tabular}{c|c|c}
		\hline
		\textbf{SID} & \textbf{Transaction} 	 & \textbf{\textit{TU}} \\ \hline
		$ S_{1} $ & (Male     Young    Doctor)  & \$20  \\ \hline
		$ S_{2} $ & (Female   Middle   Lawyer)  & \$0   \\ \hline
		$ S_{3} $ & (Male     Child    Driver)  & \$40  \\ \hline
		$ S_{4} $ & (Male     Young    Writer)  & \$12  \\ \hline
		$ S_{5} $ & (Female   Old      Artist)  & \$36  \\ \hline
	\end{tabular}
	\label{table:dimension}
\end{table}

Notice that the input database of the DHUI-Miner does not have quantities. Besides, the calculation of the overall utility of an itemset in the processed database for the two problems, the existing utility mining and the addressed MDUS problem, is different. Note that the following definition is only appropriate for MDUS problem. 

\begin{definition}
	Given a dimensional database of a sequence $ t $, let $ u(S_{i}) $ denote the utility of $ t $ in transaction $ S_{i} $. Then the overall utility of an itemset $ X $ in \textit{QSD} is defined as: $u(X) = \sum_{X\subseteq S_{i} \wedge S_{i}\in QSD} u(S_{i})$.
\end{definition}

Based on the above definition, DHUI-Miner employs a depth-first search (DFS) approach to identify high-utility itemsets and generate new itemsets by pattern combinations. When generating new itemsets, DHUI-Miner calculates the utilities of itemsets to check whether they are high-utility. Then it generates new itemsets based on the results previously found. To avoid the problem of multiple database scans, we propose a vertical data structure named \textit{utility-list} in DHUI-Miner. Note that the \textit{utility-list} is different from the original concept in \cite{liu2012mining}.

\begin{definition}[\textbf{utility-list}]
  A \textit{utility-list} of an itemset consists of a head and some records. The head part contains the \textit{\underline{name}} and the total utility of the itemset in \textit{QSD} (denoted as \textit{\underline{sutil}}). The records consist of two parts: (1) \textit{\underline{SID}}: the identifier of a transaction containing the special itemset; (2) \textit{\underline{util}}: the utility of itemsets in each transaction, which is the same as the utility of the sequence in the transaction. 
\end{definition}

\begin{figure}[!t]
	\centering
	\includegraphics[width=2.50in]{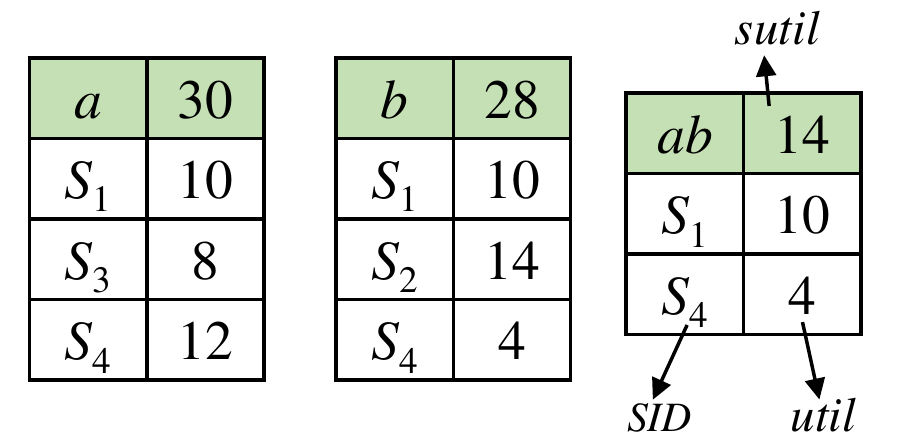}
	\caption{Utility-lists of itemsets (\textit{a}), (\textit{b}) and (\textit{ab}).}
	\label{fig-utility-list} 
\end{figure}

Fig. \ref{fig-utility-list} depicts the utility-lists of 1-itemsets (\textit{a}) and (\textit{b}), and the utility-list of 2-itemset (\textit{ab}). To generate the utility-lists of 1-itemsets, the algorithm scans the database only once to find the transactions containing each 1-itemset, and calculates the utilities of that itemset in each transaction. The utility-lists of \textit{k}-itemsets ($ k \geq 2 $) are generated by intersecting two utility-lists. The intersection of two utility-lists is the intersection of two sets. The DHUI-Miner algorithm first combines the items in the two itemsets to generate a new itemset. Then, it traverses the records in the utility-lists and adds the transaction whose identifier appears in the two utility-lists. The utility of the new itemset in a transaction is the minimum of the utilities of the two itemsets in that transaction.

Based on the construction and intersection of utility-lists, DHUI-Miner can search for all itemsets by calculating their utilities, and then discover the high-utility itemsets. As the length of itemsets increases, the number of records in utility-lists will decrease. However, exploring the complete search space has a high complexity. Therefore, we further propose a new \textit{downward closure} property, which helps to reduce the number of considered itemsets.

\begin{theorem}
	Given a multi-dimensional database \textit{QSD}, let $ Y^{k} $ be a length-$k$ pattern/itemset and $ Y^{k+1} $ be a superset of $ Y^{k} $. Then, $ u(Y^{k}) \geq u(Y^{k+1}) $.
\end{theorem}

\begin{proof}
	Since $ Y^{k+1} $ is a superset of $ Y^{k} $,  transactions containing $ Y^{k+1} $  also contain $ Y^{k} $. 
	We can get that $ u(Y^{k})$  = $\sum_{Y^{k}\subseteq S_{i} \wedge S_{i}\in QSD} u(S_{i})$ $\geq \sum_{Y^{k+1}\subseteq S_{i} \wedge S_{i}\in QSD} u(S_{i})$ = $u(Y^{k+1})$.
\end{proof}

\begin{corollary}
	\label{corollary-sd-dcp}
	\rm If an itemset $ Y^{k} $ is not a high-utility itemset, which means that $ u(Y^{k}) < u(QSD) \times \delta $, then any superset $ Y^{k+1} $ of $ Y^{k} $ is not the desired high-utility itemset since  $ u(Y^{k+1}) \leq u(Y^{k})$.
\end{corollary}

Hence, Corollary \ref{corollary-sd-dcp} can be used to reduce the number of candidates and the search space. The MDUS$ _\text{SD} $ algorithm consists of two phases. In the first phase, it employs the \textbf{search} procedure (c.f. Algorithm \ref{HUSPM-Miner}) to discover high-utility sequential patterns from the sequential part of the database. In the second phase, it applies the designed DHUI-Miner algorithm for each HUSP $ t $ to find the corresponding dimensions. The pseudo-code of DHUI-Miner is given in Algorithm \ref{DHUI-Miner}. For each 1-itemset in the dimensional database of $ t $, DHUI-Miner first scans the database to find transactions containing the itemset, and calculates the utility of the itemset in each transaction. Having this information, DHUI-Miner builds the utility-list of each 1-itemset and calculates its utility. Based on the \textit{downward closure} property, only the itemsets whose utilities are no less than $ u$(\textit{QSD}) $ \times $ $\delta $ need to be explored. After that, DHUI-Miner invokes a \textbf{DMiner} procedure to recursively search for supersets of the candidates. Finally, the MDUS$ _\text{SD} $ algorithm combines the sequential pattern $ t $ with high-utility itemsets to obtain multi-dimensional HUSPs.


\begin{algorithm}
	\caption{DHUI-Miner}
	\label{DHUI-Miner}
	\begin{algorithmic}[1]
		\REQUIRE {$ t $, high-utility sequential pattern; \textit{QSD}, a dimensional database of $ t $; $ \delta $, a minimum utility threshold.} 
		\ENSURE {The set of \textit{mdHUSPs}.}
		\STATE $ \textit{HUIs} \gets \varnothing $;
		\FOR {each $ i_{j}\in QSD $}
		\STATE $ i_{j}.UL \gets$ the utility-list of $i_{j} $;
		\IF {$ i_{j}.UL.sutil \geq u(QSD) \times \delta $}
		\STATE $ \textit{HUIs} \gets \textit{HUIs} \cup i_{j} $;
		\STATE $ exULs \gets exULs \cup i_{j}.UL$;
		\ENDIF
		\ENDFOR
		\STATE call \textbf{DMiner($ exULs, \delta, \textit{HUIs} $)};
		\STATE \textit{mdHUSPs} $ \gets \varnothing $;
		\FOR {each $ X \in HUIs $}
		\STATE \textit{mdHUSPs} $ \gets $ \textit{mdHUSPs} $\cup$ \textbf{combine($ X, t $)};
		\ENDFOR
		\STATE \textbf{return} \textit{mdHUSPs}
	\end{algorithmic}
\end{algorithm}


As shown in Algorithm \ref{DMiner}, the \textbf{DMiner} procedure first traverses the utility-lists in \textit{ULs} and combines the current utility-list with the next one to obtain a new itemset and its utility-list. If the utility of the new itemset is no less than $ u(QSD) \times \delta $, this itemset is added to the set of high-utility itemsets (\textit{HUIs}). Based on the downward closure property, its utility-list is added to the candidate set \textit{ULs} as well. Then, \textbf{DMiner}  with the candidate set as input is recursively performed. After all of the utility-lists in \textit{ULs} have been processed, \textbf{DMiner} returns the set of \textit{HUIs}. Note that the \textbf{Construct} procedure performs the intersection of two utility-lists, including combining two itemsets to obtain a new itemset and intersecting two utility-lists to obtain a new utility-list.

\begin{algorithm}
	\caption{DMiner}
	\label{DMiner}
	\begin{algorithmic}[1]
		\REQUIRE {$ ULs $, the set of utility-lists; $ \delta $, a minimum utility threshold.} 
		\ENSURE {The set of \textit{HUIs}.}
		\FOR {each $ X\in ULs $}
		\STATE $ exULs \gets \varnothing $;
		\FOR {each $ Y\in ULs$ after $ X $}
		\STATE $ Z.UL \gets $ Construct($ X, Y $);
		\IF {$ Z.UL.sutil \geq u(QSD) \times \delta $}
		\STATE $ \textit{HUIs} \gets \textit{HUIs} \cup Z $;
		\STATE $ exULs \gets exULs \cup Z.UL $;
		\ENDIF
		\ENDFOR
		\STATE call \textbf{DMiner($ exULs, \delta, \textit{HUIs} $)}.
		\ENDFOR
	\end{algorithmic}
\end{algorithm}

\begin{algorithm}
	\caption{\textbf{Judge}(\textit{prefix'}, \textit{PD}(\textit{prefix}), $ HUSPs $)}
	\label{judge}
	\begin{algorithmic}[1]
		\STATE \textit{PD}(\textit{prefix'})$\gets$\{the UL-list of $s|$\textit{prefix'}$\subseteq s\wedge s\in$ \textit{PD}(\textit{prefix});
		\STATE calculate $ u $(\textit{prefix'}) and $ PEU $(\textit{prefix'});
		\IF {$ PEU $(\textit{prefix'})$\geq \delta \times u(QSD)$}
		\IF {$ u $(\textit{prefix'})$\geq \delta \times u(QSD) $}
		\STATE $ HUSPs $$\gets$$ HUSPs $$\cup$\textit{prefix'};
		\ENDIF
		\STATE call \textbf{PGrowth(\textit{prefix'}, \textit{PD}(\textit{prefix'}), \textit{HUSPs})}.
		\ENDIF
	\end{algorithmic}
\end{algorithm}

\textbf{\underline{Complexity Analysis}}. The complexity analysis of the MDUS$ _\text{SD} $ algorithm includes two parts. The time complexity of the mining of sequence part is the same as the adopted HUSPM algorithm. The complexity analysis of dimension part is as follows. Suppose the database \textit{QSD} has $n$ transactions, and the number of dimensions is $m$. Let the number of different items (dimension values) and itemsets with length $k$ is $|I|$ and  $|L_{k}|$, respectively. The DHUI-Miner algorithm first scans the database to find the candidate 1-itemsets, which takes a time of $O(n \times m \times log(|I|))$. Generating a set of ($k$+1)-itemsets from a set of $k$-itemsets results in $|L_{k}| \times |L_{k}|$ candidates. The worst case for computing the utility value of each itemset needs $O(n)$, and this is related to the number of transactions that contain these itemsets. Therefore, the worst-case scenario requires the time complexity as $O(n \times m \times log(|I|) + \sum\nolimits_{k=1}^{m-1} |L_{k}| \times |L_{k}| \times n)$. As the length of the itemset grows, the number of the transactions that include this itemset decreases. Thus, the time to compute its utility value decreases. In addition, due to the effect of pruning strategy, the running time will be further reduced.

\section{Experiments}\label{sec:experiments} 
In this section, experiments are described to evaluate the effectiveness and efficiency of the proposed MDUS framework and the MDUS$ _\text{EM} $ and MDUS$ _\text{SD} $ algorithms.  Notice that some related algorithms (e.g., association rule mining \cite{agrawal1994fast}, itemset mining \cite{han2004mining}, sequential pattern mining algorithms \cite{pei2004mining,zaki2001spade}, utility-driven mining algorithms \cite{yin2012uspan}, etc.) are released on our Website - SPMF\footnote{\url{http://www.philippe-fournier-viger.com/spmf}} which is a widely used open-source library. However, all of these existing algorithms are not suitable compared in this paper. The reason is that this is the first study for the problem of utility mining across multi-dimensional sequences, and this addressed task is different from previous studies. Different mining task leads to different results, as well as different metrics. Therefore, none existing approach can be regarded as the baseline for comparison reasonably. All the existing studies cannot solve the MUDS problem by utilizing the time-dependent sequence-order, quantitative information, utility factor and auxiliary dimension. 

\subsection{Experimental Settings and Datasets}

Experiments were conducted on six real datasets (yoochoose, UK-online, Bible, BMS, MSNBC and Kosarak10k) and one synthetic large-scale dataset (C8S6T4I3D$ | $X$ | $K). The basic statistics of these datasets are summarized below.

$\bullet$ \textbf{yoochoose}\footnote{\url{https://recsys.acm.org/recsys15/challenge/}}: it comprises the buy events of the users over the items. The used yoochoose-buys dataset is trained on 6 months of data, containing 7,966,257 sessions of 31,637,239 clicks on 37,483 items.  It contains the sequence of clicking events when the user is surfing on the e-commerce site. Each clicking event is a time-stamp of a click on an item/product, and each record/line has the following fields/dimensions: Session ID, Timestamp, Item ID, Price, Quantity.

$\bullet$ \textbf{UK-online}\footnote{\url{http://archive.ics.uci.edu/ml/datasets/Online+Retail/}}: it contains 541,909 transactions, which occurs between 01/12/2010 and 09/12/2011 for a UK-based and registered non-store online retail. The company mainly sells unique all-occasion gifts. Thus, the original data has many noise values; it requires the pre-processing step to refine the dataset. It has 8 attribute/dimensional information, including InvoiceNo, StockCode, Description, Quantity, InvoiceDate, UnitPrice, CustomerID and Country.

$\bullet$ \textbf{Bible}, \textbf{BMS} and \textbf{MSNBC}\footnote{\url{http://www.philippe-fournier-viger.com/spmf}}: \textbf{Bible} is a real-life dataset obtained by converting the Bible, and it has 36,369 transactions with 13,905 items. Besides, the maximal number of items in sequence is 100. \textbf{BMS}  is a real-life dataset of click-stream data which collected from an e-commerce site. It has 59,601 transactions with 497 items, and the maximal number of items in sequence is 267. \textbf{MSNBC} has click-stream data obtained from the UCI repository, where the shortest sequences are removed. Total 31,790 transactions and 17 items contained in this dataset.

$\bullet$ \textbf{Kosarak10k}: this is a real-life dataset of click-stream data from a Hungarian news portal, which is a subset of kosarak\footnote{\url{http://fimi.ua.ac.be/data/}}. It has 10,000 transactions with 10,094 items, and the average number of items in each sequence is 8.1.

$\bullet$ \textbf{C8S6T4I3D$ | $X$ | $K}: it is generated by the public IBM Quest Dataset Generator \cite{Agrawal1994-dataset}. Notice that total 100,000 to 500,000 transactions are contained in this synthetic dataset. As a large-scale dataset, it is used to evaluate the scalability in this paper.

For the synthetic dataset, a simulation model \cite{tseng2013efficient,yin2012uspan,liu2012mining} was adopted to generate quantities and profit values of items. The quantity of each item was randomly generated in the range of 1 to 5. A log-normal distribution was used to randomly assign the profit values of items in the range of 0.01 to 1000.00. And the dimensions of synthetic datasets are generated using a random model. 

\textbf{\underline{Implementation Details}}: To test the proposed algorithm, a series of performance studies were conducted.  All the algorithms were implemented in Java using JDK 1.7. The experiments were executed on a personal computer equipped with an Intel Core2 i7-4790 CPU and 8 GB of RAM, running the 64-bit Microsoft Windows 7 operating system. Note that all the source code will be released on the SPMF website after the acceptance of this paper.

\subsection{Effectiveness Analysis of Patterns}
Firstly, the number of candidates and actual discovered multi-dimensional HUSPs are evaluated. The goal is to analyze the effectiveness of the proposed algorithms and the effect of different mining strategies. Note that \#${P}$1 and \#${P}$2 respectively denote the number of candidates generated by the MDUS$ _\text{EM} $ and MDUS$ _\text{SD} $ algorithms, and \#\textit{mdHUSPs} denote the number of the final discovered multi-dimensional HUSPs. The pattern results under various minimum utility thresholds are shown in Table \ref{table:candidate1}. Details of the varied $\delta$ are shown as the x-axis in Fig. \ref{fig-runtime1}.

From Table \ref{table:candidate1}, it can be seen that \#${P}$2 is either a lot better or close to \#${P}$1, while \#\textit{mdHUSPs} is always the smaller. For example on the MSNBC dataset, $\delta$ is varied from 0.01 to 0.05 at Table \ref{table:candidate1}(e). When $\delta$ is set at 0.03, the results are \#${P}$1: 506,942, \#${P}$2: 48,341, and \#\textit{mdHUSPs}: 23,224. Thus, \#${P}$2 is closed to the final results of mdHUSPs. \#${P}$2 is not  always smaller than \#${P}$1, for example on BMS dataset under all cases. These results show that a huge number of unpromising candidates were generated and determined for mining the desired multi-dimensional HUSPs. More specifically, the results of the number of generated candidates imply that the MDUS$ _\text{SD} $ algorithm has a powerful ability to filter lots of unpromising patterns. For example, when $\delta$: 0.01 on MSNBC dataset, the set that (\#${P}$1 - \#${P}$2) = 62,494,856 - 4,777,996 = 57,716,860 is successfully filtered and that \#${P}$2: 4,777,996 is highly close to \#\textit{mdHUSPs}: 4,255,916.

Experiments conducted on six real-life datasets show that the newly formulated MDUS model works in case of interactive utility mining on multi-dimensional sequential database. Besides, the results of generated patterns indicate that the parameter $\delta$ is sensitive to the MDUS framework. When one or more dimensions of information is mined while randomly setting the number of dimensions, MDUS is specific to a particular domain or application type, and can give us the informative and useful patterns.

\begin{table*}[!htbp]
	\fontsize{8.50pt}{9pt}\selectfont
	\centering
	\caption{Number of patterns by varying $\delta$.}
	\label{table:candidate1}
	\begin{tabular}{c|c|ccccc}
		\hline\hline
		\multirow{2}*{\textbf{Dataset}}&
		\multirow{2}*{\textbf{Pattern}}
		&\multicolumn{5}{c}{\textbf{Minimum utility threshold}}\\
		\cline{3-7}
		&& $\delta_1$ &  $\delta_2$  &  $\delta_3$  &   $\delta_4$   &  $\delta_5$   \\ \hline

		&  \textbf{\#${P}$1} & 3,349,156  &  2,336,030  &  1,612,493  &  1,164,948  &  872,762  \\
($a$)  yoochoose &  \textbf{\#${P}$2} & 3,041,690  &   2,121,456  &   1,142,036  &   894,123  &   621,567  \\
&\textbf{\# \textit{mdHUSPs}} & 498,315  &  245,891  &  91,288  &  39,398  &   31,549	\\
\hline

		&  \textbf{\#${P}$1} & 647,023  &  501,897  &  420,238  &  310,239  &  260,475  \\
($b$)  UK-online &  \textbf{\#${P}$2} & 622,450  &   487,918  &   402,388  &   301,003  &   243,561  \\
&\textbf{\# \textit{mdHUSPs}} & 81,764  &  75,980  &  51,078  &  31,653  &  19,666 	\\
\hline

		&  \textbf{\#${P}$1} & 19,458,716 &  5,823,863 &  1,971,661 &  1,083,898 &  595,556  \\
         ($c$)  Bible &  \textbf{\#${P}$2} & 1,484,097 &  481,196 &  190,600 &  113,786 &  72,898  \\
        &\textbf{\# \textit{mdHUSPs}} & 1,183,823 &  326,779 &  96,588 &  53,225 &  28,595 	\\
  \hline

		&  \textbf{\#${P}$1} & 116,803,878  &   7,915,800 &   4,069,994 &  2,407,121 &  1,377,358  \\
($d$)  BMS &  \textbf{\#${P}$2} & 116,804,859  &   7,916,751 &   4,070,918 &   2,407,995 &   1,378,217  \\
&\textbf{\# \textit{mdHUSPs}} & 203,966 &  14,666 &  753 &  39 &  6 	\\
\hline

		&  \textbf{\#${P}$1} & 62,494,856 &  4,199,079 &  506,942 &  218,485 &  37,673 \\
($e$)  MSNBC &  \textbf{\#${P}$2} & 4,777,996 &  302,586 &  48,341 &  22,161 &  8,866  \\
&\textbf{\# \textit{mdHUSPs}} & 4,255,916 &  234,921 &  23,224 &  9,848 &  1,389	\\
\hline

&  \textbf{\#${P}$1} & 154,247,023  &  68,100,609  &  31,304,545  &  15,129,904  &  8,022,708  \\
($f$)  Kosarak10k &  \textbf{\#${P}$2} & 154,728,144  &   67,619,206  &   30,296,237  &   13,851,331  &   6,616,889  \\
&\textbf{\# \textit{mdHUSPs}} & 317,648  &  195,980  &  159,078  &  145,653  &  139,647 	\\
\hline

		\hline\hline
	\end{tabular}
\end{table*}

\begin{figure*}[!t]
	\centering
	\includegraphics[trim=5 65 10 0,clip,scale=0.68]{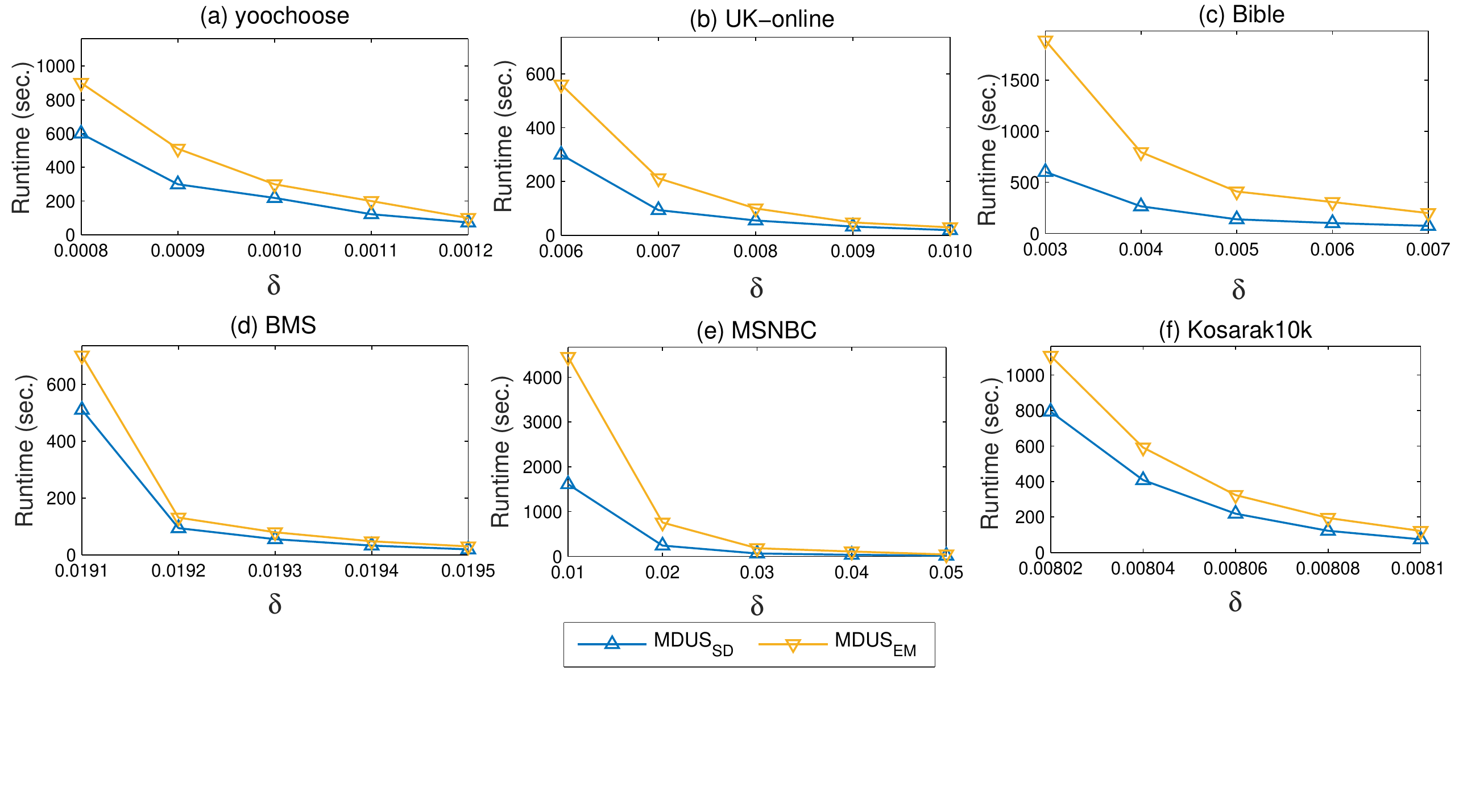}
	\caption{Runtime by varying $\delta$.}
	\label{fig-runtime1}
\end{figure*}

\begin{figure*}[!t]
	\centering
	\includegraphics[trim=5 65 20 0,clip,scale=0.68]{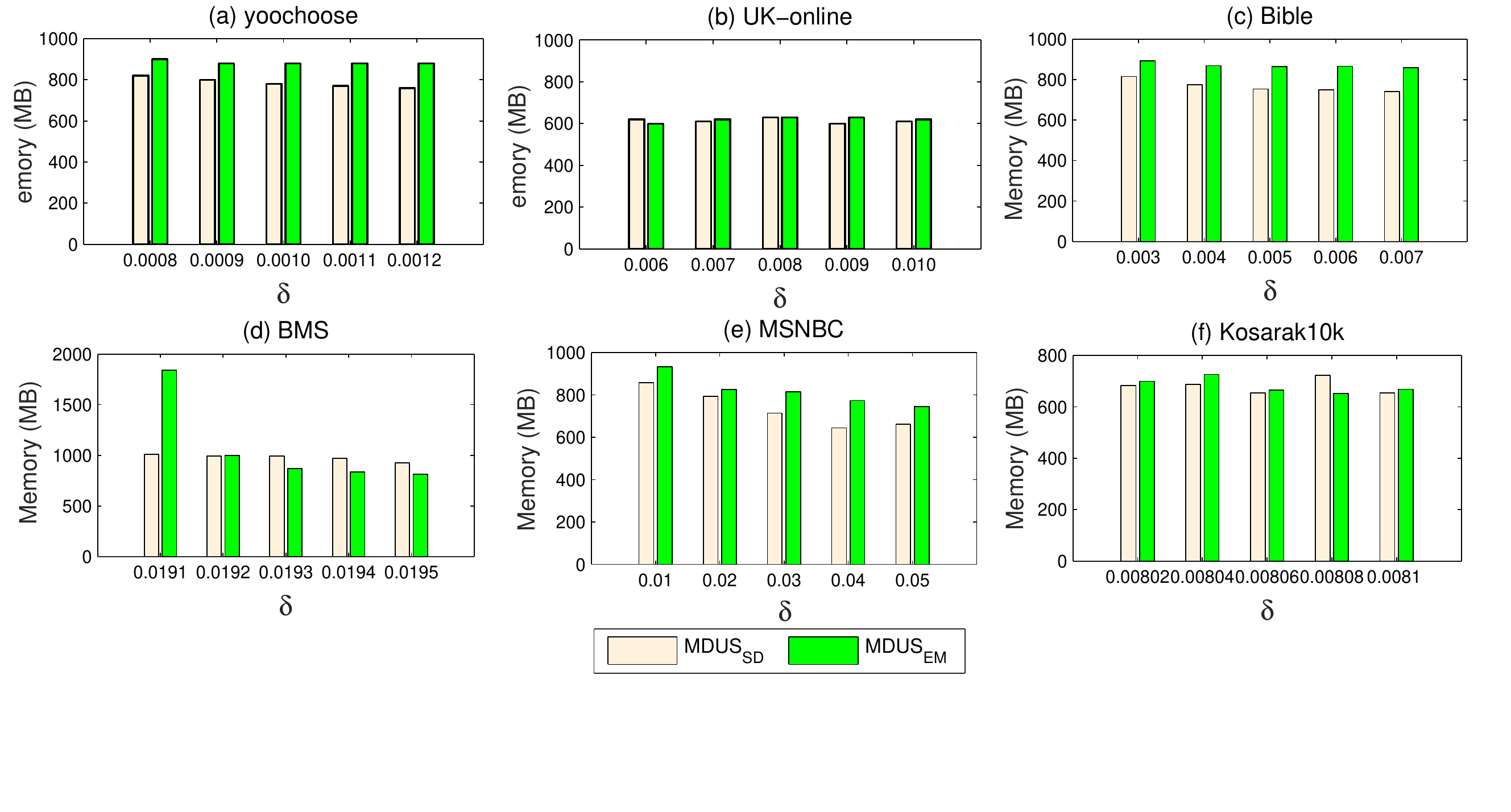}
	\caption{Memory consumption by varying $\delta$.}
	\label{fig-memory1}
\end{figure*}

\subsection{Efficiency Analysis of Runtime}
With the same parameters setting in last subsection, the efficiency w.r.t. runtime of the proposed approaches was compared under various minimum utility thresholds. Results of the runtime are reported in Fig. \ref{fig-runtime1}.

From Fig. \ref{fig-runtime1}, it can be seen that the MDUS$ _\text{SD} $ algorithm performs better than MDUS$ _\text{EM} $ for all datasets. As the minimum utility threshold decreases, the runtime of the two algorithms increases. The difference is that the runtime of MDUS$ _\text{EM} $ increases faster than that of MDUS$ _\text{SD} $, as shown in each runtime with respect to the first value of parameter $\delta$ in x-axis. It is reasonable since the proposed pruning strategy and downward closure property of DHUI-Miner can help to improve the efficiency of MDUS$ _\text{SD} $. In addition, MDUS$ _\text{SD} $ which utilizes the proposed utility-list structure can improve the performance by reducing the number of candidates and the search space. Clearly, with increasing the minimum utility threshold $\delta$, the proposed MDUS$ _\text{EM} $ approach may not return the results within a reasonable time, but the enhanced MDUS$ _\text{SD} $ algorithm can achieve a good performance in terms of runtime.

It is interesting to notice that, in Fig. \ref{fig-runtime1}(d), the runtime of the two algorithms is close to each other. The reason is that if there are always long transactions in a dataset (e.g., BMS dataset), the majority of the runtime is spent to mine the sequential part of the multi-dimensional dataset.

\subsection{Efficiency Analysis of Memory Usage}

In this subsection, the memory consumption is compared for the two algorithms. Results are shown in Fig. \ref{fig-memory1}, where it can be observed that the memory consumption of the two algorithms is similar. But in most cases, the memory consumption of MDUS$ _\text{SD} $ is less than that of MDUS$ _\text{EM} $. For example, on datasets yoochoose and MSNBC, the consumed memory of MDUS$ _\text{SD} $ is always slightly less than that of MDUS$ _\text{EM} $ under various minimum utility thresholds. This observation, to some extent, reflects the results of running time, which can be seen in Table \ref{table:candidate1} and Fig. \ref{fig-runtime1}. This is very interesting and reassuring. The reason is that the number of candidates is calculated differently in the two algorithms. MDUS$ _\text{SD} $ utilizes the pattern join, utility-list, and DHUI-Miner to search the desired patterns. The designed utility-list can be used to avoid the problem of multiple database scans. While MDUS$ _\text{EM} $ first performs the database transformation and then adopts search procedure (c.f. Algorithm 2, which may generate many candidates) to extract mdHUSPs.

More specifically, the number of candidates depends on how it is measured by the proposed algorithms. Note that the set of candidates of MDUS$ _\text{SD} $ includes candidates found in both the sequential part and the dimensional part of the dataset, while the set of candidates of MDUS$ _\text{EM} $ is the candidates of the transformed dataset. Besides, the runtime and memory consumption to process each candidate is much less for MDUS$ _\text{SD} $ than for MDUS$ _\text{EM} $. Thus, the MDUS$ _\text{SD} $ requires less runtime and memory to find the complete set of multi-dimensional high-utility sequential patterns compared to MDUS$ _\text{EM} $.

\subsection{Scalability}

Finally, we study the scalability of the MDUS$ _\text{SD}$ and MDUS$ _\text{EM} $ algorithms by using two measurements: speedup of runtime and memory efficiency. Experiments were conducted on the large-scale synthetic C8S6T4I3D$ | $X$ | $K dataset by changing transactions from 100,000 to 500,000 with the embedding 50 dimensions. In all scalability experiments, we use two parameter settings: the minimum utility threshold  $\delta$: 0.0005 and $\delta$: 0.0008.  Fig. \ref{fig-scalability}(a) to Fig. \ref{fig-scalability}(f) report the results in terms of running time, memory usage and the number of patterns (candidates and final results) with respect to different number of transactions in test dataset. 

\begin{figure*}[!t]
	\centering 
	\includegraphics[trim=20 100 20 0,clip,scale=0.56]{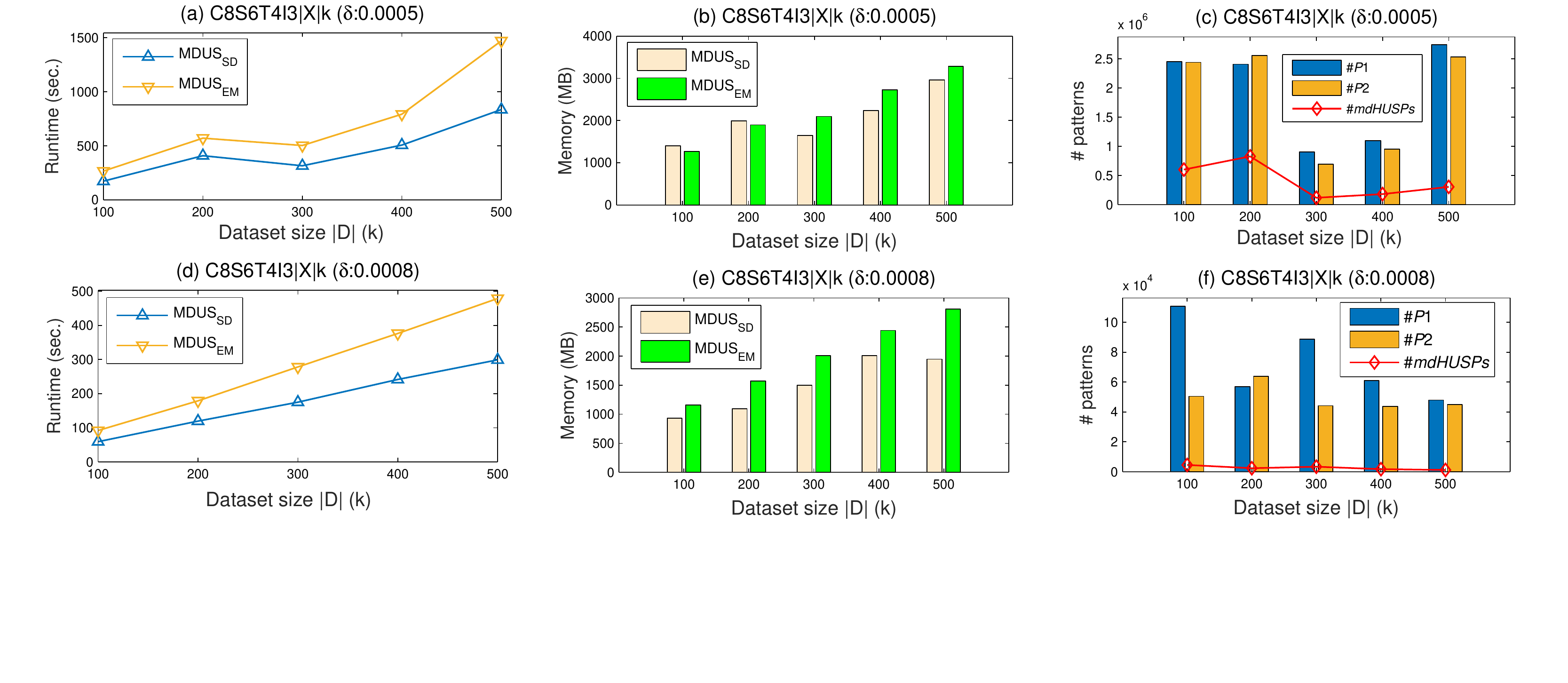}
	\caption{Scalability of the compared algorithms.}
	\label{fig-scalability}
\end{figure*}

Observed from Fig. \ref{fig-scalability}, we can see that proposed technique is highly scalable. It can be clearly seen that the scalability of MDUS$ _\text{SD} $ is always better than that of MDUS$ _\text{EM} $ on large-scale dataset under different parameter settings. Handling the sequential part and dimensional part of the database separately, the MDUS$ _\text{SD} $ algorithm always consumes less runtime and memory than the equivalence transformation-based MDUS$ _\text{EM} $. As the dataset size increases, both the runtime and memory consumption of the two algorithms increase, and the gap on runtime between the compared algorithms increases significantly. But the gap of memory consumption does not always increases as shown in Fig. \ref{fig-scalability}(b) and Fig. \ref{fig-scalability}(e). The reason is that the value $u(QSD) \times \delta$ increases as the size of the dataset is increased. Thus, the number of generated candidates satisfying the threshold may decrease, and this may decrease the memory consumption. In summary, the second developed  MDUS$ _\text{SD}  $ algorithm outperforms the first MDUS$ _\text{EM} $ algorithm in terms of runtime, memory consumption and scalability.

\section{Conclusions and Future Work}
\label{sec:conclusion}

In this paper, we first formulate the problem of utility maximization across multi-dimensional sequences, which having rich features including sequence-order information, utility factor, and the dimension information. We propose a novel framework named MDUS to extract \textbf{\underline{M}}ulti-\textbf{\underline{D}}imensional \textbf{\underline{U}}tility-oriented \textbf{\underline{S}}equential useful patterns. By utilizing the utility theory, MDUS can discover high utility sequential patterns from multi-dimensional sequential databases. In addition, two algorithms named MDUS$ _\text{EM} $ and MDUS$ _\text{SD} $ are further proposed for MDUS with the different data structures and filtering strategies. Extensive experiments have been done on the real-world datasets, and the results have demonstrated the effectiveness and efficiency of the MDUS framework. More specifically, MDUS$ _\text{SD} $ outperforms MDUS$ _\text{EM} $ in terms of runtime, memory consumption and scalability. 

In the future, several advanced issues can be considered such as discover the top-$k$ patterns using MDUS, apply the MDUS framework to deal with big data or other tasks.

\bibliographystyle{IEEEtran}
\bibliography{main}



\vspace{-1cm}
\begin{IEEEbiography}[{\includegraphics[width=1in,height=1.25in,clip,keepaspectratio]{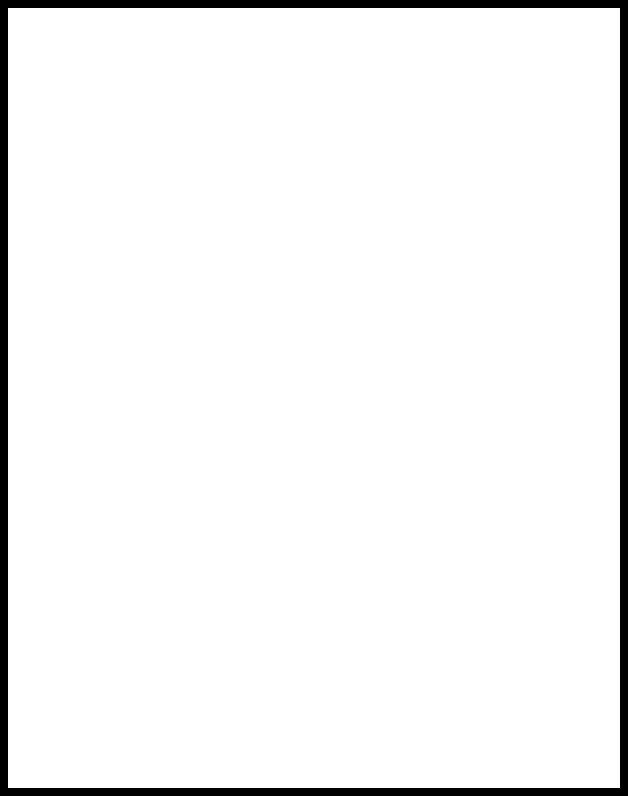}}]{Wensheng Gan} received the Ph.D. in Computer Science and Technology, Harbin Institute of Technology (Shenzhen), Guangdong, China in 2019. He received the B.S. degree in Computer Science from South China Normal University, Guangdong, China in 2013. His research interests include data mining, utility computing, and big data analytics. He has published more than 50 research papers in  peer-reviewed journals and international conferences, which have received more than 450 citations.
\end{IEEEbiography}

\vspace{-1cm}
\begin{IEEEbiography}[{\includegraphics[width=1in,height=1.25in,clip,keepaspectratio]{newAuthor.png}}]{Jerry Chun-Wei Lin}
	is an associate professor at Western Norway University of Applied Sciences, Bergen, Norway. He received the Ph.D. in Computer Science and Information Engineering, National Cheng Kung University, Tainan, Taiwan in 2010. His research interests include data mining, big data analytics, and social network. He has published more than 250 research papers in peer-reviewed international conferences and journals, which have received more than 2500 citations. He is the co-leader of the popular SPMF open-source data mining library and the Editor-in-Chief (EiC) of the \textit{Data Mining and Pattern Recognition} (DSPR) journal, and Associate Editor of \textit{Journal of Internet Technology}. 
\end{IEEEbiography}

\vspace{-1cm}
\begin{IEEEbiography}[{\includegraphics[width=1in,height=1.25in,clip,keepaspectratio]{newAuthor.png}}]{Jiexiong Zhang}
	is currently a senior software engineer in Didi Chuxing, Beijing, China. He received the M.S. degrees in Computer Science from Harbin Institute of Technology (Shenzhen), Guangdong, China in 2017. His research interests include data mining, artificial intelligence, and big data analytics. 
\end{IEEEbiography}

\vspace{-1cm}
\begin{IEEEbiography}[{\includegraphics[width=1in,height=1.25in,clip,keepaspectratio]{newAuthor.png}}]{Hongzhi Yin}
	is a senior lecturer and ARC DECRA Fellow with the University of Queensland, Australia. He received the PhD degree in computer science from Peking University in 2014. His research interests include recommender system, user profiling, topic models, deep learning, social media mining, and location-based services. He has published over 90 papers in the most prestigious journals (i.e., TKDE, TCYB, ACM TOIS, ACM TIST, ACM TKDD) and conferences (i.e., SIGMOD, KDD, VLDB, ICDE, WSDM). Besides, he has one monograph published by Springer.
	
\end{IEEEbiography}

\vspace{-1cm}
\begin{IEEEbiography}[{\includegraphics[width=1in,height=1.25in,clip,keepaspectratio]{newAuthor.png}}]{Philippe Fournier-Viger}
	is full professor and Youth 1000 scholar at the Harbin Institute of Technology (Shenzhen), Shenzhen, China. He received a Ph.D. in  Computer Science at the University of Quebec in Montreal (2010). His research interests include pattern mining, sequence analysis and prediction, and social network mining. He has published more than 200 research papers in refereed international conferences and journals. He is the founder of the popular SPMF open-source data mining library, which has been cited in more than 700 research papers. He is Editor-in-Chief (EiC) of the \textit{Data Mining and Pattern Recognition} (DSPR) journal.
\end{IEEEbiography}

\vspace{-1cm}
\begin{IEEEbiography}[{\includegraphics[width=1in,height=1.25in,clip,keepaspectratio]{newAuthor.png}}]{Han-Chieh Chao}
	has been the president of National Dong Hwa University since February 2016. He received M.S. and Ph.D. degrees in Electrical Engineering from Purdue University in 1989 and 1993, respectively. His research interests include high-speed networks, wireless networks, IPv6-based networks, and artificial intelligence. He has published nearly 500 peer-reviewed professional research papers. He is the Editor-in-Chief (EiC) of IET Networks and \textit{Journal of Internet Technology}. Dr. Chao has served as a guest editor for ACM MONET, IEEE JSAC, \textit{IEEE Communications Magazine}, \textit{IEEE Systems Journal}, \textit{Computer Communications}, \textit{IEEE Proceedings Communications}, \textit{Wireless Personal Communications}, and \textit{Wireless Communications \& Mobile Computing}. Dr. Chao is an IEEE Senior Member and a fellow of IET. 
\end{IEEEbiography}

\vspace{-2.5cm}
\begin{IEEEbiography}[{\includegraphics[width=1in,height=1.25in,clip,keepaspectratio]{newAuthor.png}}]{Philip S. Yu }
	received the B.S. degree in electrical engineering from National Taiwan University, M.S. and Ph.D. degrees in electrical engineering from Stanford University, and an MBA from New York University. He is a distinguished professor of computer science with the University of Illinois at Chicago (UIC) and also holds the Wexler Chair in Information Technology at UIC. Before joining UIC, he was with IBM, where he was manager of the Software Tools and Techniques Department at the Thomas J. Watson Research Center. His research interests include data mining, data streams, databases, and privacy. He has published more than 1,300 papers in peer-reviewed journals (i.e., TKDE, TKDD, VLDBJ, ACM TIST) and conferences (KDD, ICDE, WWW, AAAI, SIGIR, ICML, etc). He holds or has applied for more than 300 U.S. patents. Dr. Yu was the Editor-in-Chief of \textit{ACM Transactions on Knowledge Discovery from Data}. He received the ACM SIGKDD 2016 Innovation Award, and the IEEE Computer Society 2013 Technical Achievement Award. Dr. Yu is a fellow of the ACM and the IEEE.
\end{IEEEbiography}

\end{document}